\documentclass[final,3p]{elsarticle}

\usepackage{amssymb}
\usepackage{amsmath}
\usepackage{amsthm}

\usepackage{pdfpages}
\usepackage{mathtools}
\usepackage{thmtools}
\usepackage{svg}
\usepackage{float}
\usepackage{hyperref}
\usepackage{comment}
\usepackage{enumerate}
\usepackage{graphicx}
\usepackage{multirow}
\usepackage{subcaption}
\usepackage{todonotes}
\usepackage{xcolor}
\usepackage{array}
\usepackage{makecell}
\usepackage{fullpage}
\usepackage{setspace}

\usepackage{tikz}
\usetikzlibrary{arrows, bending}

\usetikzlibrary{shapes,calc,arrows.meta,intersections,positioning}
\tikzset{
	position/.style args={#1:#2 from #3}{
		at=(#3.#1), anchor=#1+180, shift=(#1:#2)
	}
}

\newcolumntype{x}[1]{>{\centering\arraybackslash}p{#1}}

\newcommand{\SequenceSwapSP}{\textsc{SemiSwap}}
\newcommand{\SequenceSwapSPv}{$\textsc{SemiSwap}(v)$}

\newcommand{\MaxAssets}{\textsc{MaxAssets}}
\newcommand{\MaxAssetsv}{$\textsc{MaxAssets}_v$}

\newcommand{\Reach}{\textsc{Reach}}
\newcommand{\Reachv}{$\textsc{Reach}(v)$}
\newcommand{\Reachvl}{$\textsc{Reach}(v,\ell)$}

\newcommand{\Swap}{\textit{swap}}

\newcommand{\maxSat}{\textsc{Max-2-SAT}}
\newcommand{\setCover}{\textsc{Set Cover}}
\newcommand{\vertexCover}{\textsc{Vertex Cover}}
\newcommand{\partition}{3-\textsc{Partition}}
\newcommand{\Tpartition}{\textsc{Partition}}
\newcommand{\IS}{\textsc{Independent Set}}
\newcommand{\satcon}{3-\textsc{SAT Connectivity}}

\newcommand{\bsigma}{\boldsymbol{\sigma}}

\newcommand{\deb}{\textsf{de}}
\newcommand{\cre}{\textsf{cr}}
\newcommand{\fnetwork}{\mathcal{F} = (G, \deb, \cre, \vecl, \veca^x, \vecf)}
\newcommand{\svo}{s_{v_1}}
\newcommand{\svt}{s_{v_2}}
\newcommand{\tvo}{t_{v_1}}
\newcommand{\tvt}{t_{v_2}}
\newcommand{\aovo}{a_{v_1}}
\newcommand{\aovt}{a_{v_2}}
\newcommand{\atvo}{\sig{a}_{v_1}}
\newcommand{\atvt}{\sig{a}_{v_2}}

\newcommand{\vecf}{\mathbf{f}}
\newcommand{\vecl}{\mathbf{l}}
\newcommand{\veca}{\mathbf{a}}
\newcommand{\vecp}{\mathbf{p}}
\newcommand{\vecc}{\mathbf{c}}
\newcommand{\F}{\mathcal{F}}
\newcommand{\G}{\mathcal{G}}
\newcommand{\HH}{\mathcal{H}}
\newcommand{\sig}[1]{{#1}^{\sigma}}
\newcommand{\pre}[1]{{#1}^{\text{pre}}}
\newcommand{\diff}[1]{{#1}^{\text{diff}}}

\newcommand{\N}{\mathcal{N}}

\newcommand{\NN}[0]{\mathbb{N}}
\newcommand{\RR}[0]{\mathbb{R}}

\newcommand{\classP}{\textsf{P}}
\newcommand{\classNP}{\textsf{NP}}
\newcommand{\classPLS}{\textsf{PLS}}
\newcommand{\classAPX}{\textsf{APX}}
\newcommand{\classPSPACE}{\textsf{PSPACE}}

\newtheorem{theorem}{Theorem}[section]

\newtheorem{proposition}[theorem]{Proposition}
\newtheorem{lemma}[theorem]{Lemma}
\newtheorem{corollary}[theorem]{Corollary}
\newtheorem{observation}[theorem]{Observation}

\newtheorem{example}[theorem]{Example}
\newtheorem{definition}[theorem]{Definition}

\newtheorem*{openproblem*}{Open Problem}

\newcommand{\changed}[1]{#1}%{\color{green!60!black} #1}}

\renewcommand{\paragraph}[1]{\medskip \noindent \textrm{\textbf{#1.}} $\,$}

\journal{Theoretical Computer Science}

\begin{document}

\begin{frontmatter}
\title{Dynamic Debt Swapping in Financial Networks}

%% Author name
\author[ffm]{Henri Froese}
\ead{henri.froese@stud.uni-frankfurt.de}
\author[rwth]{Martin Hoefer}
\ead{mhoefer@cs.rwth-aachen.de}
\author[rwth]{Lisa Wilhelmi\corref{cor1}}
\ead{wilhelmi@algo.rwth-aachen.de}
\cortext[cor1]{Corresponding author}

%% Author affiliation
\affiliation[ffm]{organization={Dept.\ of Computer Science, Goethe University Frankfurt},%Department and Organization
            addressline={Robert-Mayer-Strasse 11-15}, 
            city={Frankfurt am Main},
            postcode={60325}, 
            %state={},
            country={Germany}}
\affiliation[rwth]{organization={Dept.\ of Computer Science, RWTH Aachen University},%Department and Organization
            addressline={Ahornstrasse 55}, 
            city={Aachen},
            postcode={52074}, 
            %state={},
            country={Germany}}

%% Abstract
\begin{abstract}
A \emph{debt swap} is an elementary edge swap in a directed, weighted graph, where two edges with the same weight swap their targets. Debt swaps are a natural and appealing operation in financial networks, in which nodes are banks and edges represent debt contracts. They can improve the clearing payments and the stability of these networks. However, their algorithmic properties are not well-understood.

We analyze the computational complexity of debt swapping. Our main interest lies in \emph{semi-positive} swaps, in which no creditor strictly suffers and at least one strictly profits. These swaps lead to a Pareto-improvement in the entire network. We consider network optimization via sequences of $v$-improving debt swaps from which a given bank $v$ strictly profits. For ranking-based clearing, we show that every sequence of semi-positive $v$-improving swaps has polynomial length. In contrast, for arbitrary $v$-improving swaps, the problem of reaching a network configuration that allows no further swaps is \classPLS-complete.
        
In global optimization, the goal is to maximize the utility of a given bank $v$ by performing a sequence of debt swaps in the network. This problem is \classNP-hard to approximate for multiple types of swaps.

Moreover, we study reachability problems -- deciding if a sequence of swaps exists between given initial and final networks. We design a polynomial-time algorithm to decide this question for arbitrary swaps and derive hardness results for several other types of swaps.

Many of our results can be extended to networks with arbitrary monotone clearing.
\end{abstract}


\end{frontmatter}

\clearpage

\section{Introduction}

The international financial system represents a complex and highly interconnected network. A prominent threat is that tight structural connections between banks, funds, and other financial institutions can yield devastating cascading effects. These systemic risks have been a major issue in the financial crisis of 2008 and several other crises since then. Mitigating these risks is an important problem, which has recently also received attention in computer science. The goal here is to understand the computational complexity of network-based interbank models for financial systems and potential interventions. 

A fundamental aspect of modeling financial systems is \emph{clearing}, a process of debt resolution by which solvency of banks in the system gets determined. In a financial network (where banks are nodes, and weighted, directed edges correspond to debt contracts) the clearing problem asks for a network flow that represents a consistent assignment of payments of debtor banks to their creditors. Based on the clearing, a bank might not receive the full amount of credit it gave to others. As a consequence, it might not have sufficient funds to clear all her own debt. Naturally, the solution to the clearing problem is of central importance to banks since it directly affects their assets.

Banks have an obvious incentive to improve their clearing conditions. Moreover, govern\-men\-tal regulators, central banks, and other financial services share an interest in maintaining and improving the financial system. While allocation rules for clearing are often fixed, other interventions for improvement have the potential to significantly improve the clearing conditions.
A particularly interesting variant of network improvement is \emph{debt swapping}~\cite{PappDebtSwapping}: There are two debt contracts among two pairs of banks $(u_1,v_1)$ and $(u_2,v_2)$. The contracts have the same value. In a debt swap, the creditor banks $v_1$ and $v_2$ change their role as recipients of these contracts. This leaves every bank with the same amounts of debt as before -- only the network is affected by the exchange of the edges. A debt swap is a local edge-swap operation and represents a small change in the network structure. Since it involves only a small number of banks, it can be executed in a distributed fashion by the participating banks. Banks are more likely to agree to this operation if it does not deteriorate their network position, i.e., it does not decrease their assets in the clearing of the resulting network.

Debt swaps are a very natural and interesting swap operation with many conceptual advantages. For example, they are usually less demanding than the more prominent approach of portfolio compression~\cite{SchuldenzuckerPortfolioCompression, veraart2020PortfolioCompression} (for a discussion of the relation between these two operations, see~\cite{PappDebtSwapping}). Moreover, seemingly different operations proposed by regulators (such as, e.g., claims trades) can be interpreted as debt swaps~\cite{HoeferVW24}. From a computational perspective, debt swaps are related to a substantial body of literature on edge swap operations in graphs. Numerous variants of such swaps have attracted a large interest in computer science over the decades (see related work discussed below). The computational aspects of debt swapping, however, are not well understood.

In this paper, we strive to understand debt swap dynamics and their effects on the network. Since a debt swap is a local improvement step in the network, we first focus on local optimization (e.g., from an initial network, can we reach a configuration that allows no more profitable debt swaps?). In addition, we also consider aspects of global optimization (i.e., maximize the assets of a single bank via a sequence of debt swaps) and reachability problems (i.e., decide if a target network structure can be reached using a sequence of debt swaps). We discuss all our results consistently for networks in which banks use natural ranking-based clearing payments~\cite{CsokaDecentralizedClearing,BertschingerStrategicPayments,ioannidis2022financial}. Indeed, many of our results can be extended to arbitrary monotone clearing. Throughout the paper we mention the extension explicitly whenever it applies.

\paragraph{Contribution and Techniques}
After discussing preliminaries in Section~\ref{sec:model}, we classify debt swaps into positive (both creditors profit strictly in the clearing after the swap), semi-positive (one creditor profits strictly, the other one remains the same), or arbitrary ones. Our characterization results in Section~\ref{sec:characterization} show that in every financial network, there can be no positive swap. Semi-positive swaps can exist, and, in fact, they always represent a \emph{Pareto-improvement}: After the swap \emph{no bank in the entire network suffers}, every bank obtains at least as many assets as before. As such, the \emph{recovery rate} of every bank weakly increases, as well as the payments towards every debt contract. This makes semi-positive swaps very attractive. Similar results were shown by Papp and Wattenhofer~\cite{PappDebtSwapping} for proportional payments. Our paper generalizes these results substantially to \emph{all monotone payments}.

We then concentrate on improvement dynamics with semi-positive debt swaps. Since they Pareto-improve the assets of all banks (and strictly for a creditor bank), every sequence of semi-positive swaps is inherently acyclic and terminates in a local optimum. A natural and challenging question that we study in Section~\ref{sec:debt-swap-dyn} is whether there are short sequences that allow to reach a local optimum in polynomial time. \changed{Towards answering this question, we} further refine the characterization of semi-positive swaps for ranking-based payments. \changed{Our characterization involves several classes of swaps. For each but one class, we show that there can be only a polynomial number of swaps from that class in every sequence. The existence of short sequences then depends crucially on one remaining class of so-called \emph{active extension swaps} (for a formal statement see Definitions~\ref{def:extension-swap} and~\ref{def:active-extension-swap} below).}

When a single bank $v$ strictly profits from every semi-positive swap, then we prove that the number of active extension swaps is also bounded by a polynomial in the size of the network (and, thus, so is every sequence of semi-positive swaps). Instead, when we consider arbitrary swaps in which a given bank $v$ strictly improves, then there are initial networks from which \emph{every} sequence of such swaps has exponential length. More fundamentally, we prove that it is \classPLS-complete to construct any local optimum, i.e., any network where liabilities of each node are consistent with the initial network and that allows no (arbitrary) debt swap in which the assets of a given bank $v$ strictly increase. 

The problem without the requirement of strict improvement for a single bank $v$ turns out to be extremely challenging. We also make some progress here -- under a condition on the liabilities and the in-degree of partly paid edges of every bank, we can show existence of a short sequence of semi-positive swaps to a local optimum. The general problem remains as an interesting avenue for future work.

In Section~\ref{sec:opt}, we consider global optimization, i.e., the problem of determining the maximal assets of a given bank that can be obtained by a sequence of debt swaps from a given initial network. We show that this problem is \classNP-hard to approximate, both when considering arbitrary or only semi-positive debt swaps. Similar hardness results apply when determining the maximal \emph{sum of assets of several banks}. 

In Section~\ref{sec:reach} we then turn to reachability problems: Given an initial and a target network on the same set of nodes $V$, the goal is to decide whether or not there is a sequence of debt swaps to turn the initial network into the target network. For arbitrary debt swaps, we show that the problem can be solved in polynomial time. In contrast, when we have to maintain a lower bound on the assets of a bank $v$ throughout, the problem becomes \classPSPACE-complete. For semi-positive debt swaps, we show that the problem is \classNP-hard. 
	
In Table~\ref{table:summary} we summarize our main results with references to the respective theorems and corollaries in the technical part.

\begin{table}[h!]
\renewcommand{\arraystretch}{1.3}
\begin{center}
\changed{
\begin{tabular}{l||c|c}
 & \multicolumn{2}{c}{Local Optimization}        \\
Swaps         &           & $v$-improving        \\\hline\hline
Semi-Positive & \classP\ in special cases & \classP \\
              & (Thm. \ref{thm:shortSequence}) & (Thm. \ref{thm:shortSequenceV}) \\ \hline
Arbitrary     &           & \classPLS-complete  \\
              &           & (Thm. \ref{thm:PLS}) \\
\multicolumn{3}{c}{}\\
\multicolumn{3}{c}{}\\
              & \multicolumn{2}{c}{Global Optimization} \\
Swaps         &            &$v$-improving \\ \hline\hline
Semi-Positive & $n^{1/2-\varepsilon}$-hard & \classNP-hard   \\ 
              & (Thm. \ref{thm:opt-IS}), Cor. \ref{cor:opt-IS}) & (Thm. \ref{thm:opt-2partition}, Thm. \ref{thm:opt-3partition}) \\ \hline
Arbitrary     & \classP  & \classPSPACE-complete \\
& (Thm. \ref{thm:greedy})  & (Cor. \ref{cor:PSPACE}/Thm. \ref{thm:PSPACE}) \\
\multicolumn{3}{c}{}\\
\multicolumn{3}{c}{}\\
              & \multicolumn{2}{c}{Reachability}   \\
Swap          &       & $v$-improving/with LB \\ \hline\hline
Semi-Positive & \multicolumn{2}{c}{\classNP-hard} \\
              & \multicolumn{2}{c}{(Prop. \ref{prop:reach-NP})} \\ \hline
Arbitrary     & \classAPX-hard  & \classAPX-hard \\
              & (Cor. \ref{cor:maxAssetHard}) & (Cor. \ref{cor:maxAssetImproveHard}) \\
\end{tabular}
}
\caption{Main results of this paper\label{table:summary}}
\end{center}
\end{table}

\paragraph{Related Work}
There has been a recent surge of interest in the algorithmic and game-theoretic aspects of financial networks. Most works adopt and extend the clearing model by Eisenberg and Noe~\cite{EisenbergSystemicRisk} (discussed in more detail below), where banks are nodes, edges are debt contracts, and clearing is governed by proportional payments. Most closely related to our work is recent work by Papp and Wattenhofer~\cite{PappDebtSwapping}, who introduce and study the concept of debt swaps for proportional payments. They show that positive swaps do not exist, semi-positive swaps lead to Pareto-improvement, discuss shock models and hardness of a number of optimization problems. Another network-altering operation was recently introduced by Kanellopoulos et al.~\cite{kanellopoulos2023debt}. A debt transfer describes the process of a bank transferring a debt claim to another creditor. Formally, this operation exchanges a path of length two by one single edge. The authors consider optimization problems and a game-theoretic model where the banks strategically decide whether or not to perform a debt transfer. 

More generally, algorithmic aspects of (extensions of) the Eisenberg-Noe model have been subject to substantial recent interest, including, e.g., the computation of clearing states in networks with credit-default swaps (CDS)~\cite{schuldenzucker2017clearing, ioannidis2022financial, IKV_ICALP22}. Some clearing models (e.g., using CDS'es) suffer from default ambiguity~\cite{SchuldenzuckerDefault}, and several strategies for resolving this problem have been proposed~\cite{PappW21_WINE}. Changes of the instance in terms of adding or deleting single liabilities, donations to other banks, lending, manipulation of external assets, allocation of stimulus checks, or bailouts were analyzed in~\cite{PappNetworkAware,kanellopoulos2022forgivingDebt, egressy2024price, papachristou2022allocating}. While most works assume the clearing state to manifest as a fixed point, the impact of \emph{sequential} updates is studied in~\cite{PappW21}. Further temporal aspects are investigated in~\cite{FriedetzkyKMST25}, where liabilities are associated with individual maturity intervals. The works investigate the complexity of allocating integral and fractional payments in order to minimize defaults.

There is a growing interest in clearing problems beyond the restriction to proportional payments, using more general classes of monotone functions to model the clearing process~\cite{CsokaDecentralizedClearing}. This extension poses many interesting algorithmic and, in particular, game-theoretic problems that are very recently starting to get explored~\cite{BertschingerStrategicPayments,KanellopoulosNetworkGames,HoeferWilhelmi}. 
%\LW{gekürzt}

An extended abstract of this work has been \changed{published in the proceedings of the 4th Symposium on Algorithmic Foundations of Dynamic Networks (SAND 2025)~\cite{FroeseHW25}.}

%%%%%%%%%%%%%%%%%%%%%%%%%%%%%%%%%%%%%%%%%%%%%%%%%%%%%%%%%%%%%%%%%%%%%%%%%%%%%%%%%%%%%%%%%%%%%%%%%%%%%%%%%%%%%%%%%%%%%
\section{Model and Preliminaries}
\label{sec:model}

\paragraph{Financial Networks} 
A \emph{financial network} $\fnetwork$ %\F=(G, \vecl, \veca^x, \vecf)$ 
is given by a weighted, directed multi-graph $G=(V,E)$ with $|V|=n$ nodes and $|E|=m$ edges. Each node $v \in V$ resembles a financial institution or \emph{bank}. Each edge $e \in E$ represents a debt contract, where a debtor bank $\deb(e) \in V$ owes an amount of $l_e \in \NN_{>0}$ to a creditor bank $\cre(e)$. Formally, $e$ is a directed edge $(\deb(e), \cre(e))$ with weight $l_e$. We assume that there are no loops, i.e., $\deb(e) \neq \cre(e)$. The set of incoming and outgoing edges of $v$ are denoted by $E^-(v) = \{ e \in E \mid \cre(e) = v\}$ and $E^+(v) = \{e \in E \mid \deb(e) = v\}$. The \emph{total liabilities} of a bank $v$ are defined as $L_v=\sum_{e \in E^+(v)} l_e$. Further, every bank holds \emph{external assets} $a^x_v \in \NN$, i.e., assets of $v$ that are independent of the network. The payments that a bank receives from its debtors within the network are called \emph{internal assets} and, together with external assets, form the available assets of a bank.

An \emph{allocation rule} is a collection $\vecf=(\vecf_v)_{v\in V}$, where for bank $v$ $\vecf_v = (f_e)_{e \in E^+(v)}$ determines a distribution of the available assets $b_v \in \RR_{\geq 0}$ towards her liabilities. $f_e : \RR_{\geq 0} \rightarrow [0, l_e]$ is the \emph{payment function} for edge $e \in E^+(v)$. For every bank $v$, the payment functions satisfy the following properties: (1) $f_e$ is continuous\footnote{Monotonicity, edge capacity, and limited flow conservation together imply continuity. We postulate it explicitly for convenience.} and non-decreasing in $b_v$ (monotonicity), (2) each bank pays at most the respective liability, i.e., $f_e(b_v) \leq l_e$ (edge capacity), and (3) every bank spends all assets until all liabilities are paid off, i.e., $\sum_{e \in E^+(v)}f_e(b_v) = \min \{b_v, L_v\}$ (limited flow conservation). The available assets of $v$ are defined as the sum of external and internal assets $b_v = a^x_v + \sum_{e \in E^-(v)}f_e(b_{\deb(e)}) \ge 0$. 

These constraints give rise to a fixed-point problem. A vector of assets $(b_v)_{v \in V}$ that fulfills all constraints is called \emph{feasible}. For every allocation rule $\vecf$, the set of feasible vectors forms a complete lattice w.r.t.\ coordinate-wise comparison~\cite{BertschingerStrategicPayments, CsokaDecentralizedClearing}. Consequently, for every allocation rule $\vecf$ there exists a supremum $(a_v)_{v \in V}$ of feasible assets and resulting payments $\vecp =(p_e)_{e\in E}$ with $p_e = f_{e}(a_{\deb(e)})$ that pointwise maximize the assets of each bank and the payments on each edge. \changed{Following previous literature, we assume} the clearing settles on these maximal payments, which we term the \emph{clearing state}. 
\changed{Thus, the clearing state is uniquely determined for every instance.}
In the clearing state, bank $v$ receives \emph{total incoming assets} of $\sum_{e \in E^-(v)}p_e$ and holds \emph{total assets} of $a_v = a^x_v + \sum_{e \in E^-(v)}p_e$.

In this paper, we will discuss all our results consistently in the context of a natural class of \emph{edge-ranking} rules~\cite{BertschingerStrategicPayments, CsokaDecentralizedClearing,ioannidis2022financial}: For each bank $v$, her outgoing edges $E^+(v)$ are ordered by a permutation $\pi_v$. First, $v$ directs all her payments towards edge $\pi_v(1)$ until the edge is completely paid off or $v$ has no remaining assets. Then, she pays off edge $\pi_v(2)$ until the edge is completely paid off or $v$ has no remaining assets. This process continues until either all debt is settled or $v$ runs out of funds. 

Apart from being natural and well-motivated in applications, edge-ranking rules also have appealing computational properties. Given a financial network with edge-ranking rules, a maximal clearing state can be computed in strongly polynomial time~\cite{BertschingerStrategicPayments}. Moreover, these favorable properties extend to general monotone integral~\cite{CsokaDecentralizedClearing,%} (also termed unit-ranking~\cite{
BertschingerStrategicPayments} %)
allocation rules in explicit representation by using edge-rankings over a suitable set of auxiliary edges~\cite{BertschingerStrategicPayments}. In this way, \emph{all our statements generalize} to unit-ranking rules with edge-ranking representation.

Other classes of payments have also been of interest in the literature (e.g., \emph{proportional payments}, where every bank allocates her total assets in proportion to her total liabilities, or combinations of ranking and proportional). Indeed, many of our results extend even to arbitrary monotone allocation rules.

\paragraph{Debt Swaps} 
For a given financial network $\F$, a debt swap $\sigma$ is a structural change to the network. For two edges $e_1, e_2$ among four distinct banks with the same liability (i.e., $l_{e_1} = l_{e_2} =: \ell$) we exchange their creditors. The result can be described by an adjustment in the creditor function by swapping the entries of $e_1$ and $e_2$. Note that the allocation rule $\vecf$ of the banks remains unaffected by the swap, since the debtor function remains unchanged. However, due to the structural change, a debt swap may change the payments in the clearing state. We denote the resulting financial network, creditor function, clearing state, and total assets of a bank $v$ after the swap by $\sig{\F}$, $\sig{\cre}$, $\sig{\vecp}$ and $\sig{a}_v$, respectively.

\begin{definition}[Debt Swap]
Consider a financial network $\F$ with edges $e_1, e_2 \in E$. We use the short notation $v_1 = \cre(e_1)$, $v_2 = \cre(e_2)$, $u_1 = \deb(e_1)$ and $u_2 = \deb(e_2)$ throughout. Suppose that the liabilities $l_{e_1} = l_{e_2}$ and $u_1,u_2,v_1,v_2$ are pairwise distinct nodes. A \emph{debt swap} $\sigma$ of $e_1$ and $e_2$ creates a new network $\sig{\F}=(G,\deb,\sig{\cre},\vecl,\veca^x,\vecf)$ where $\sig{\cre}(e_1) = v_2$, $\sig{\cre}(e_2) = v_1$, and $\sig{\cre}(e) = \cre(e)$ for all $e \in E \setminus \{e_1,e_2\}$. 
\end{definition}
\changed{Throughout the paper, $e_1$ and $e_2$ are used to refer to the trading edges, while $u_1, u_2, v_1$ and $v_2$ refer to the respective debtor and creditor banks.}
Observe that a debt swap may introduce multi-edges even in initially simple graphs.
\changed{The edge-ranking rule of a bank can be interpreted as a priority ranking of contracts (instead of the creditor banks). When performing a debt swap of two edges $e_1$ and $e_2$, we assume that $u_1$ allocates payments towards the new edge with creditor $v_2$ in the same way as for $v_1$ before the trade. Hence, besides exchanging the creditor of one outgoing edge, the permutation of outgoing edges (and thus the priority ranking) remains the same.}

We proceed with a small example that shows some of the interesting effects of debt swaps. 

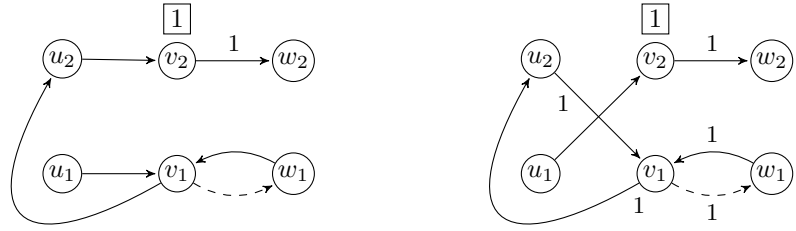
\begin{figure}[t]
\centering
\begin{tikzpicture}[>=stealth', shorten >=1pt, auto,
    node distance=1cm, %scale=0.8, 
    transform shape, align=center, 
    bank/.style={circle, draw, inner sep = 1}]
    \node[bank] (u1) at (0,0) {$u_1$};
	\node[bank] (v1) [right = of u1] {$v_1$};
	\node[bank] (u2) [above = of u1] {$u_2$};
	\node[bank] (v2) [above = of v1] {$v_2$};
	\node[bank] (w1) [right = of v1] {$w_1$};
	\node[bank] (w2) [right = of v2] {$w_2$};
	\node[rectangle, draw=black, inner sep=2.5pt] (xv2) [above=0.1cm of v2] {$1$};
	\draw[->] (u1) to (v1);
	\draw[->] (u2) to (v2);
	\draw[->] (v1) to[in=240, out=210, looseness=2.8] (u2);
	\draw[->,dashed] (v1) to[bend right] (w1);
	\draw[->] (w1) to[bend right] (v1);
	\draw[->] (v2) to node[above, midway] {\small $1$} (w2);
\end{tikzpicture}
\hspace{1.5cm}
\begin{tikzpicture}[>=stealth', shorten >=1pt, auto,
    node distance=1cm, %scale=0.8, 
    transform shape, align=center, 
    bank/.style={circle, draw, inner sep = 1}]
    \node[bank] (u1) at (0,0) {$u_1$};
	\node[bank] (v1) [right = of u1] {$v_1$};
	\node[bank] (u2) [above = of u1] {$u_2$};
	\node[bank] (v2) [above = of v1] {$v_2$};
	\node[bank] (w1) [right = of v1] {$w_1$};
	\node[bank] (w2) [right = of v2] {$w_2$};
	\node[rectangle, draw=black, inner sep=2.5pt] (xv2) [above=0.1cm of v2] {$1$};
	\draw[->] (u1) to (v2);
	\draw[->] (u2) to node[below=2pt, pos=0.1] {\small $1$} (v1);
	\draw[->] (v1) to[in=240, out=210, looseness=2.8] node[below=2pt, at start] {\small $1$} (u2);
	\draw[->,dashed] (v1) to[bend right] node[below, midway] {\small $1$} (w1);
	\draw[->] (w1) to[bend right] node[above, midway] {\small $1$} (v1);
	\draw[->] (v2) to node[above, midway] {\small $1$} (w2);
\end{tikzpicture}
\caption{The figures visualize a financial network $\F$  before (left) and $\sig{\F}$ after a debt swap (right). All edges have unit weight and edge labels indicate payments. For $v_1$, the solid edge is the preferred outgoing edge. $v_2$ has external assets of 1 as shown in the node label.}
\label{fig:introEx}
\end{figure}

\begin{example}\label{ex:intro}\rm
Consider the initial financial network in Fig.~\ref{fig:introEx} (left) with edge-ranking allocation rules. Bank $v_1$ first pays off all debt to $u_2$ before directing payments to $w_1$, i.e, $\pi_{v_1}(1)=(v_1,u_2), \pi_{v_1}(2)=(v_1,w_1)$. All other banks pay towards their unique outgoing edge. Consider the clearing state. For bank $v_2$ the external assets of 1 equal her total liabilities $L_{v_2}$. Thus, $v_2$ can settle all her debt towards $w_2$. For the other edges, suppose there are positive payments on edge $(u_2,v_2)$. Then, by definition, $u_2$ has to have these assets available, so $v_1$ must pay this amount to $u_2$. By the same argument, $u_1$ must have positive total assets and pay them to $v_1$. However, $u_1$ obviously has no assets, and there are no payments on these edges. Due to the allocation rule of $v_1$, there are no payments in the cycle of $v_1$ and $w_1$, unless $v_1$ fully pays off $(v_1,u_2)$ first (for which she needs assets from $u_1$). The clearing state has a payment of 1 on $(v_2,w_2)$. The total assets are 1 for $v_2$ and $w_2$, and 0 for all others. 

Now consider the financial network in Fig.~\ref{fig:introEx} (right) after the swap of edges $(u_1,v_1)$ and $(u_2,v_2)$ to $(u_1,v_2)$ and $(u_2,v_1)$. Obviously, $v_2$ can still pay off her single outgoing edge. After the swap, there is a cycle including $v_1$ and $u_2$ over prioritized edges. When the payments on both edges are raised to 1, the fixed-point properties remain satisfied. Since $v_1$ pays off all debt on the edge with highest priority, she will direct any additional assets towards the edge of her second priority to $w_1$. Here a new cycle with $w_1$ over prioritized edges arises, and the fixed-point conditions allow to further raise payments along that cycle to 1. 

The additional payments after the swap result from the emergence of the cycle between $u_2$ and $v_1$. The total assets of $v_1,u_2$ and $w_1$ are increased to $a_{v_1}=2, a_{u_2}=1$ and $a_{w_1}=1$ after the swap. For the total assets of the other banks the swap is inconsequential. 
\hfill $\blacksquare$
\end{example}

A debt swap of two edges $e_1$ and $e_2$ maintains the local network properties of the debtors $u_1$ and $u_2$. It has no effect on the set of liability values. Moreover, it does not impact their incoming edges. As such, we are interested in the gains and losses in total assets of the creditors $v_1$ and $v_2$, since (arguably) they are the key actors to implement a debt swap. The following classification of debt swaps slightly extends the one in~\cite{PappDebtSwapping}. Observe that the debt swap explained in Example~\ref{ex:intro} is semi-positive and Pareto-improving.

\begin{definition}[Swap Types]\label{def:types-of-swaps}
A debt swap $\sigma$ of edges $e_1$ and $e_2$ is called (1) \emph{positive} if $\sig{a}_{v_1} > a_{v_1}$ and $\sig{a}_{v_2} > a_{v_2}$, (2) \emph{semi-positive} if $\sig{a}_{v_1} \geq a_{v_1}$ and $\sig{a}_{v_2} \geq a_{v_2}$ and exactly one in\-equality is strict, (3) \emph{Pareto-improving} if $\sig{a}_{v} \geq a_{v}$ for all $v\in V$ and at least one inequality is strict.
\end{definition}

\section{Characterization}\label{sec:characterization}
To reason about debt swaps, we first provide a structural analysis of the clearing payments before and after debt swaps. We obtain a characterization of semi-positive and positive swaps. The analysis substantially generalizes techniques from~\cite{PappDebtSwapping} and yields two major insights. First, there are no positive swaps. Second, every semi-positive debt swap yields a Pareto-improvement in the \emph{entire network}, i.e., no bank in the financial network is harmed while at least one bank strictly profits. Thus, semi-positive swaps can only improve the recovery rate of each bank (i.e., the ratio of available assets over liabilities). Consequently, the payment assigned to each debt contract does not decrease. Interestingly, the converse also holds: Every Pareto-improving debt swap is semi-positive. 

Similar to~\cite{PappDebtSwapping}, we use an auxiliary network with the same payments as the resulting network after the debt swap. In this auxiliary network, we replace the swap edges with auxiliary sources and sinks with appropriate assets representing the payments on the swap edges. For a contradiction, we assume that the debt swap is positive. By tracing the payments in the auxiliary network, we show that payments in the original network before the swap did not correspond to the maximum clearing state. This contradiction shows that 
positive swaps cannot exist. As a rather direct consequence from this characterization, we obtain that semi-positive swaps lead to a Pareto-improvement in the entire network.

The main distinction from the proofs in~\cite{PappDebtSwapping} for proportional payments is that these payments have limited structure in the auxiliary network and, as such, are easier to handle. Instead, our proofs rely only on monotonicity of the allocation rule. By arguing about a broad range of possibilities and dependencies of potential payments, we can show that these results indeed apply to all financial networks with \emph{any monotone allocation rule}. 

\subsection{Structural Properties}
Monotonicity of the allocation rule implies that when the total assets of bank $v$ increase, then payments on each $e \in E^+(v)$ do not decrease. Intuitively, therefore, every clearing state of a financial network can be decomposed into two clearing states of appropriate networks. We proceed in two steps: First, each bank $v$ is assigned reduced external assets $\hat{a}^x_v \leq a^x_v$. This results in a new financial network, for which we get the first clearing state. In this state, $v$ (partially) pays her liabilities. For the second step, we reduce each liability by the payment in the first step, and $v$ receives external assets of $a^x_v-\hat{a}^x_v$. This again results in a new financial network, for which we get the second clearing state. The sum of the two clearing states is the clearing state of the original network.

This separability condition can be used to show that increasing the external assets of one bank Pareto-improves the total assets of every bank.  

\begin{restatable}[Monotonicity]{lemma}{lemMonotonicity}
Suppose two financial networks $\F$ and $\hat{\F}$ differ only in the external assets of a single bank $v$ such that $a^x_v > \hat{a}^x_v$. For every bank $u \in V$, the total assets in the clearing state satisfy $a_u \geq \hat{a}_u$.
\end{restatable}

The payments in a financial network emerge due to a combination of external assets and cyclic money flow. Clearly, a sink node $t$ (with $E^+(t) = \emptyset$) is not included in any cycle. Thus, all incoming payments of a sink node originated from external assets of other banks. 

Indeed, we can show that the increase of $\Delta>0$ in external assets of a source node $s$ (with $E^-(s) = \emptyset$) raises the total assets of all sinks in sum by at most $\Delta$. 

\begin{restatable}[Non-Expansivity]{lemma}{lemNonExpan}\label{lem:non-expansivity}
For a given financial network, let $T\subset V$ be the set of all sink nodes. If the external assets of a source $s\in V \setminus T$ are increased to $a'^x_s = a^x_s + \Delta$, then the total assets of all sinks increase by at most $\Delta$, i.e., $\sum_{t \in T} a'_t \leq \Delta + \sum_{t \in T} a_t$.
\end{restatable}

A similar property can be shown if $s$ is not a source bank. The argument uses an auxiliary construction that externalizes the incoming payments of bank $s = v$ by splitting $v$ into two separate nodes $s_v$ and $t_v$. The sink $t_v$ holds the same incoming edges as $v$ in the original graph. Similarly, source $s_v$ has the same outgoing edges as $v$. Setting the external assets of $s_v$ to the total assets $a_v$ of $v$ in the clearing state ensures that the payments of $s_v$ and $v$ are equal. We can then verify that the modified network has the same clearing state as the original network, and, in particular, $t_v$ has the same incoming payments as $v$. Clearly, the nodes representing $v$ are not contained in any cycle.

\begin{figure}[t]
\begin{subfigure}[t]{0.32\textwidth}
   % \hspace{1.5cm}
    \centering
    \resizebox{!}{1.1\textwidth}{
    \begin{tikzpicture}[>=stealth', shorten >=1pt, auto,
    node distance=1cm, scale=1, 
    transform shape, align=center, 
    bank/.style={circle, draw, minimum size=2.5pt}]
    \node[bank] (u1) at (0,0) {$u_1$};
	\node[bank] (v1) [below = of u1] {$v_1$};
	\node[bank] (u2) [right = of u1] {$u_2$};
	\node[bank] (v2) [right = of v1] {$v_2$};
	\node[bank] (w1) [below = of v1] {$w_1$};
	\node[bank] (w2) [below = of v2] {$w_2$};
	\node[rectangle, draw=black, inner sep=2.5pt] (xv2) [right=0.1cm of v2] {$1$};
	\draw[->] (u1) to (v2);
	\draw[->] (u2) to node[below=2pt, pos=0.1] {\small $1$} (v1);
	\draw[->] (v1) to[in=120, out=150, looseness=2] node[above=2pt, midway] {\small $1$} (u2);
	\draw[->,dashed] (v1) to[bend right] node[right, midway] {\small $1$} (w1);
	\draw[->] (w1) to[bend right] node[right, midway] {\small $1$} (v1);
	\draw[->] (v2) to node[right, midway] {\small $1$} (w2);
\end{tikzpicture}}
\end{subfigure}
\begin{subfigure}[t]{0.32\textwidth}
    \centering
    %\hspace{-2cm}
    \resizebox{!}{\textwidth}{
    \begin{tikzpicture}[>=stealth', shorten >=1pt, auto,
        node distance=0.5cm, scale=1, 
        transform shape, align=center, 
        bank/.style={circle, draw, inner sep=0.1cm}]
        \node[bank] (u1) at (0,0) {$u_1$};
    	\node[circle, draw=black, inner sep=0.05cm] (tv1) [below = of u1] {$t_{v_1}$};
    	\node[circle, draw=black, inner sep=0.05cm] (sv1) [below = of tv1] {$s_{v_1}$};
    	\node[bank] (u2) [right = of u1] {$u_2$};
    	\node[circle, draw=black, inner sep=0.05cm] (tv2) [below = of u2] {$t_{v_2}$};
    	\node[circle, draw=black, inner sep=0.05cm] (sv2) [below = of tv2] {$s_{v_2}$};
    	\node[bank] (w1) [below = of sv1] {$w_1$};
    	\node[bank] (w2) [below = of sv2] {$w_2$};
    	\node[rectangle, draw=black, inner sep=2.5pt] (xv1) [left=0.15cm of sv1] {$2$};
    	\node[rectangle, draw=black, inner sep=2.5pt] (xv2) [right=0.1cm of sv2] {$1$};
    	\draw[->] (u1) to (tv2);
    	\draw[->] (u2) to node[right, pos=0.3] {\small $1$} (tv1);
    	\draw[->] (sv1) to[in=140, out=150, looseness=1.9] node[left, midway] {\small $1$} (u2);
    	\draw[->,dashed] (sv1) to[bend right] node[right, midway] {\small $1$} (w1);
    	\draw[->] (w1) to[bend right] node[right, midway] {\small $1$} (tv1);
    	\draw[->] (sv2) to node[right, midway] {\small $1$} (w2);
    \end{tikzpicture}}
\end{subfigure}
\begin{subfigure}[t]{0.32\textwidth}
    \centering
    \resizebox{!}{\textwidth}{
    \begin{tikzpicture}[>=stealth', shorten >=1pt, auto,
    node distance=0.5cm, scale=1, 
    transform shape, align=center, 
    bank/.style={circle, draw, inner sep=0.1cm}]
    %\node (G) at (1,4) {$\sig{\mathcal{G}}$};
    \node[bank] (u1) at (0,0) {$u_1$};
    \node[circle, draw=black, inner sep=0.05cm] (tu1) [below = of u1] {$t_{u_1}$};
	\node[bank] (v1) [below = of tu1] {$v_1$};
	\node[bank] (u2) [right = of u1] {$u_2$};
	\node[circle, draw=black, inner sep=0.05cm] (tu2) [below = of u2] {$t_{u_2}$};
	\node[bank] (v2) [right = of v1] {$v_2$};
	\node[bank] (w1) [below = of v1] {$w_1$};
	\node[bank] (w2) [below = of v2] {$w_2$};
	\node[rectangle, draw=black, inner sep=2.5pt] (xv1) [left=0.15cm of v1] {$1$};
	\node[rectangle, draw=black, inner sep=2.5pt] (xv2) [right=0.1cm of v2] {$1$};
	\draw[->] (u1) to (tu1);
	\draw[->] (u2) to node[right, midway] {\small $1$} (tu2);
	\draw[->] (v1) to[in=140, out=150, looseness=1.9] node[left, midway] {\small $1$} (u2);
	\draw[->,dashed] (v1) to[bend right] node[right, midway] {\small $1$} (w1);
	\draw[->] (w1) to[bend right] node[right, midway] {\small $1$} (v1);
	\draw[->] (v2) to node[right, midway] {\small $1$} (w2);
\end{tikzpicture}}
\end{subfigure}
\caption{The network on the left corresponds to the network $\sig{\F}$ in Example~\ref{ex:intro}. Black edge labels indicate payments. The middle network results from applying source-sink equivalency to $v_1$ and $v_2$, whereas the right network results from applying source-sink equivalency for edges $(u_2,v_1)$ and $(v_2,u_1)$.}
\label{fig:source-sink-eq}
\end{figure}

\begin{definition}
    For a given financial network $\F$ with bank $v$, we define an adjusted network $\HH$ as follows. We replace $v$ by a source node $s_v$ and a sink node $t_v$. Sink $t_v$ has external assets 0 and the same set of incoming edges as $v$, i.e., $\cre_{\HH}(e) = t_v$ for every $e \in E^-(v)$ and ${E_{\HH}}^-(t_v) = E^-(v)$. Source $s_v$ has external assets $a_v$ and the same outgoing edges as $v$, i.e., $\deb_{\HH}(e) = s_v$ for every $e \in E^+(v)$, and ${E_{\HH}}^+(s_v)=E^+(v)$. Source $s_v$ in $\HH$ uses the allocation rule of $v$ in $\F$.
\end{definition}
An example construction can be seen in Fig.~\ref{fig:source-sink-eq}, where banks $v_1$ and $v_2$ are replaced by sources and sinks. 
The next lemma shows that the clearing state in the auxiliary network is equivalent to the one in the original network.

\begin{restatable}[Source-Sink-Equivalency]{lemma}{lemSourceSinkEq}\label{lem:source-sink-eq}
For a given financial network $\F$ and an adjusted network $\HH$ it holds that
\begin{enumerate}
    \item[(1)] the networks have equivalent clearing states, i.e., $p^{\HH}_e = p_e$ for all $e \in E$.
    \item[(2)] the sink $t_v$ has total assets $a_v - a^x_v$.
\end{enumerate}
\end{restatable}

A very similar construction can be used to externalize the payments on an edge. The right image in Fig.~\ref{fig:source-sink-eq} depicts an example construction.

\begin{corollary}[Source-Sink-Equivalency for Edges]\label{cor:sink-source-eq}
For a given financial network $\F$ with edge $e$, let $u = \deb(e)$ and $v = \cre(e)$. Add the sink $t_v$ and set the target of $e$ to $\cre(e) = t_v$. Increase the external assets of $v$ by $p_e$. Then, the resulting financial network $\G$ has the same clearing state as $\F$, i.e., $p^{\G}_e = p_e$, for all $e \in E$.
\end{corollary}

\subsection{Existence of Debt Swaps}\label{sec:ex}

Example~\ref{ex:intro} shows that semi-positive and Pareto-improving debt swaps can exist, for both proportional and edge-ranking rules.

\begin{observation}\label{obs:swap_existence}
    For proportional payments and edge-ranking rules, financial networks with semi-positive and Pareto-improving debt swaps exist.
\end{observation}

Our main result in this section shows that every financial network with a monotone allocation rule has no positive debt swap. 

\begin{restatable}{theorem}{thmNoPositive}\label{thm:no-pos-swap}
    A financial network with monotone allocation rules allows no positive debt swap.
\end{restatable}

\begin{figure}[t]
\begin{subfigure}[t]{0.32\textwidth}
\centering
\resizebox{!}{0.5\textwidth}{
\begin{tikzpicture}[>=stealth', shorten >=1pt, auto,
    node distance=1cm, scale=1, 
    transform shape, align=center, 
    bank/.style={circle, draw, inner sep=0.1cm}]
    \node[bank] (u1) at (0,0) {$u_1$};
	\node[bank] (v1) [below = of u1] {$v_1$};
	\node[bank] (u2) [right = of u1] {$u_2$};
	\node[bank] (v2) [below = of u2] {$v_2$};
	\node[rectangle, draw=black, inner sep=2.5pt] (xv1) [left=0.1cm of v1] {\small{$a^x_{v_1}$}};
	\node[rectangle, draw=black, inner sep=2.5pt] (xv2) [right=0.1cm of v2] {\small{$a^x_{v_2}$}};
    \draw[->] (u1) to  node[right, midway] {\small $p_{e_1}$} (v1);
	\draw[->] (u2) to node[right, midway] {\small $p_{e_2}$} (v2);
	\draw[->] (-0.6,0.6) to (u1);
	\draw[->] (0.6,0.6) to (u1);
	\draw[->] (1.1,0.6) to (u2);
	\draw[->] (2.3,0.6) to (u2);
	\draw[->] (-0.6,-1) to (v1);
	\draw[->] (0.6,-1) to (v1);
	\draw[->] (1.1,-1) to (v2);
	\draw[->] (2.3,-1) to (v2);
	\draw[->] (v1) to (-0.6,-2.4);
	\draw[->] (v1) to (0.6,-2.4);
	\draw[->] (v2) to (1.1,-2.4);
	\draw[->] (v2) to (2.3,-2.4);
\end{tikzpicture}}
\end{subfigure}
\begin{subfigure}[t]{0.32\textwidth}
\centering
\resizebox{!}{0.75\textwidth}{
\begin{tikzpicture}[>=stealth', shorten >=1pt, auto,
    node distance=1cm, scale=1, 
    transform shape, align=center, 
    bank/.style={circle, draw, inner sep=0.1cm}]
    \node[bank] (u1) at (0,0) {$u_1$};
	\node[circle, draw=black, inner sep=0.05cm] (tu1) [below = of u1] {$t_{u_1}$};
	\node[bank] (v1) [below = of tu1] {$v_1$};
	\node[bank] (u2) [right = of u1] {$u_2$};
	\node[circle, draw=black, inner sep=0.05cm] (tu2) [below = of u2] {$t_{u_2}$};
	\node[bank] (v2) [below = of tu2] {$v_2$};
	\node[rectangle, draw=black, inner sep=2.5pt] (xv1) [left=0.1cm of v1] {\small{$a^x_{v_1}+p_{e_1}$}};
	\node[rectangle, draw=black, inner sep=2.5pt] (xv2) [right=0.1cm of v2] {\small{$a^x_{v_2}+p_{e_2}$}};
    \draw[->] (u1) to  node[right, midway] {\small $p_{e_1}$} (tu1);
	\draw[->] (u2) to node[right, midway] {\small $p_{e_2}$} (tu2);
	\draw[->] (-0.6,0.6) to (u1);
	\draw[->] (0.6,0.6) to (u1);
	\draw[->] (1.1,0.6) to (u2);
	\draw[->] (2.3,0.6) to (u2);
	\draw[->] (-0.6,-2.8) to (v1);
	\draw[->] (0.6,-2.8) to (v1);
	\draw[->] (1.1,-2.8) to (v2);
	\draw[->] (2.3,-2.8) to (v2);
	\draw[->] (v1) to (-0.6,-4);
	\draw[->] (v1) to (0.6,-4);
	\draw[->] (v2) to (1.1,-4);
	\draw[->] (v2) to (2.3,-4);
\end{tikzpicture}}
\end{subfigure}
\begin{subfigure}[t]{0.32\textwidth}
\centering
\resizebox{!}{0.9\textwidth}{
\begin{tikzpicture}[>=stealth', shorten >=1pt, auto,
    node distance=0.7cm, scale=1, 
    transform shape, align=center, 
    bank/.style={circle, draw, inner sep=0.1cm}]
    \node[bank] (u1) at (0,0) {$u_1$};
	\node[circle, draw=black, inner sep=0.05cm] (tu1) [below = of u1] {$t_{u_1}$};
	\node[circle, draw=black, inner sep=0.05cm] (tv1) [below = of tu1] {$t_{v_1}$};
	\node[circle, draw=black, inner sep=0.05cm] (sv1) [below = of tv1] {$s_{v_1}$};
	\node[bank] (u2) [right = of u1] {$u_2$};
	\node[circle, draw=black, inner sep=0.05cm] (tu2) [below = of u2] {$t_{u_2}$};
	\node[circle, draw=black, inner sep=0.05cm] (tv2) [below = of tu2] {$t_{v_2}$};
	\node[circle, draw=black, inner sep=0.05cm] (sv2) [below = of tv2] {$s_{v_2}$};
	\node[rectangle, draw=black, inner sep=2.5pt] (xv1) [left=0.1cm of sv1] {\small{$a_{v_1}$}};
	\node[rectangle, draw=black, inner sep=2.5pt] (xv2) [right=0.1cm of sv2] {\small{$a_{v_2}$}};
    \draw[->] (u1) to  node[right, midway] {\small $p_{e_1}$} (tu1);
	\draw[->] (u2) to node[right, midway] {\small $p_{e_2}$} (tu2);
	\draw[->] (-0.63,0.63) to (u1);
	\draw[->] (0.63,0.63) to (u1);
	\draw[->] (0.8,0.63) to (u2);
	\draw[->] (2,0.63) to (u2);
	\draw[->] (-0.63,-2.2) to (tv1);
	\draw[->] (0.63,-2.2) to (tv1);
	\draw[->] (0.8,-2.2) to (tv2);
	\draw[->] (2,-2.2) to (tv2);
	\draw[->] (sv1) to (-0.63,-4.8);
	\draw[->] (sv1) to (0.63,-4.8);
	\draw[->] (sv2) to (0.8,-4.8);
	\draw[->] (sv2) to (2,-4.8);
\end{tikzpicture}}
\end{subfigure}
\vspace{0.2cm}\\
\begin{subfigure}[t]{0.32\textwidth}
\centering
    $\F$
\end{subfigure}
\begin{subfigure}[t]{0.32\textwidth}
\centering
    $\G$
\end{subfigure}
\begin{subfigure}[t]{0.32\textwidth}
\centering
    $\HH$
\end{subfigure}
\caption{Schematic illustration of networks $\G$ and $\HH$ for a financial network $\F$.}
\label{fig:no-positiv-swap}
\end{figure}

To prove this result, we assume for a contradiction there exists a network $\F$ with a positive swap of edges $e_1$ and $e_2$. To analyze the payments in $\F$, source-sink-equivalency is applied to the swap edges. Let $\G$ denote the resulting network, see Fig.~\ref{fig:no-positiv-swap} for an illustration. To accommodate for the reduced incoming payments $a_{v_1} - a^x_{v_1} - p_{e_1}$ of $v_1$, the external assets are set to $a^x_{v_1} + p_{e_1}$. The external assets of $v_2$ are set analogously.  

The network $\HH$ is obtained by additionally applying source-sink equivalency to $v_1$ and $v_2$ in $\G$. Then, $t_{v_1}$ and $t_{v_2}$ receive incoming payments of $a_{v_1} - a^x_{v_1} - p_{e_1}$ and $a_{v_2} - a^x_{v_2} - p_{e_2}$, respectively. This construction has no impact on the payments in the clearing state. Analogously, the networks $\sig{\G}$ and $\sig{\HH}$ are constructed for the network $\sig{\F}$ after the debt swap. The networks $\G, \sig{\G}$ and $\HH, \sig{\HH}$ have the same cycles but differ in external assets, respectively.

The next technical lemma characterizes the effects of increasing external assets of $\svo$ and $\svt$ in $\HH$.

Let $\HH_{\Delta}$ be the resulting network when increasing the external assets of $\svo$ and $\svt$ in $\HH$ by $\delta_{\svo} \in [0, \atvo - \aovo]$ and $\delta_{\svt} \in [0, \atvt - \aovt]$, respectively. We use $\Delta_{w}$ to denote the increase in total assets of bank $w$, i.e., $\Delta_w = a^{\HH_{\Delta}}_w - a_w$. We use $\F_{\Delta}$ to denote the network when reversing the source-sink equivalency both for $t_{v_1}, s_{v_1}$ and $t_{v_2}, s_{v_2}$ as well as for $t_{u_1}$ and $t_{u_2}$.

\begin{restatable}{lemma}{lemNoPosSwapTwo}\label{lem:no-pos-swap2}
For $\HH_{\Delta}$ and $\F_{\Delta}$ the following are true:
\begin{enumerate}
    \item[(1)] $\delta_{\svo} + \delta_{\svt} = \Delta_{\tvo} + \Delta_{\tvt} + \Delta_{t_{u_1}} + \Delta_{t_{u_2}}$.
    \item[(2)] If $\delta_{\svo}> 0$ or $\delta_{\svt}>0$, then $\delta_{\svo} \neq \Delta_{\tvo} + \Delta_{t_{u_1}}$ and $\delta_{\svt} \neq \Delta_{\tvt} + \Delta_{t_{u_2}}$.
\end{enumerate}
\end{restatable}

\smallskip
Suppose there exists a positive debt swap for a financial network $\F$, in which the total assets of both creditor banks $v_1$ and $v_2$ strictly increase by $\delta_{v_1}>0$ and $\delta_{v_2}>0$. In $\HH$ the debt swap is simulated by increasing the total assets of source nodes $\svo$ and $\svt$ by $\delta_{v_1}$ and $\delta_{v_2}$, respectively. We will show that there exists a suitable set of values for $\delta_{v_1}$ and $\delta_{v_2}$ such that $v_1$ receives $\delta_{v_1}$ additional incoming payments. Intuitively, then, the pairs $u_1$, $v_1$ and $u_2, v_2$ are each contained in the same cycle. Hence, payments could be raised along the cycles, which contradicts maximality of the clearing state. This idea can be used to complete the proof of Theorem~\ref{thm:no-pos-swap}.

The theorem shows that both creditors participating in a debt swap cannot simultaneously strictly increase their total assets. In the next part, we consider the payments on the edges $e_1$ and $e_2$ before and after a generic swap $\sigma$. The first proposition shows they cannot both be strictly increased.

\begin{restatable}{proposition}{propNoSwapPosPaym}\label{prop:no-swap-pos-payments}
    For every financial network, there is no debt swap $\sigma$ such that $\sig{p}_{e_1} > p_{e_1}$ and $\sig{p}_{e_2} > p_{e_2}$.
\end{restatable}

\begin{restatable}{proposition}{propSemiSwapPaym}\label{prop:semi-swap-payments}
    For every semi-positive debt swap $\sigma$, we have either $\sig{p}_{e_1} > p_{e_1}$ and $\sig{p}_{e_2} \leq p_{e_2}$, or $\sig{p}_{e_2} > p_{e_2}$ and $\sig{p}_{e_1} \leq p_{e_1}$.
\end{restatable}
\smallskip

Towards semi-positive swaps, the total assets of $v_1$ and $v_2$ Pareto-improve. Perhaps surprisingly, the second main result of this section shows that a swap $\sigma$ is semi-positive if and only if it Pareto-improves the assets of every bank in the network. In particular, no bank in the financial network is harmed while at least one bank strictly profits. The proof of the theorem uses our previous results on source-sink-equivalency and an auxiliary network $\HH$.

\begin{restatable}{theorem}{thmPareto}\label{thm:semi-swap-eq-Pareto-impro}
  For every financial network with monotone allocation rules, a debt swap is semi-positive if and only if it is Pareto-improving.
\end{restatable}

\changed{For edge-ranking payments it is crucial that the ranking of debtors remains the same, even if, say, $u_1$ already had another existing edge to $v_2$ before the swap. Suppose that after the swap, the swap edge would get merged with the existing one to the better position in the ranking. Then creditors of edges ranked in between could suffer from lower payments after the swap. Such a swap with subsequent merge could be semi-positive without being Pareto-improving. In contrast, for proportional payments a subsequent merge of edges is unproblematic, since the payments towards each creditor are determined by the overall proportion of debt but not the splitting into individual edges.}

%%%%%%%%%%%%%%%%%%%%%%%%%%%%%%%%%%%%%%%%%%%%%%%%%%%%%%%%%%%%%%%%%%%%%%%%%%%%%%%%%%%%%%%%%%%%%%%%%%%%%%%%%%%%%%%%%%%%%
\section{Debt Swap Dynamics}\label{sec:debt-swap-dyn}
\label{sec:local}
When two creditors perform a debt swap, the structure of the underlying graph changes. This might enable further swaps which potentially further increase the assets of the involved banks. In this way, the swap operation naturally gives rise to \emph{dynamics} where, in each step, a pair of creditors agree to perform a swap in the current network. Conceptually, semi-positive swaps are very attractive -- they are equivalent to Pareto-improving swaps and improve the situation not only for the creditors but for the entire network. 

In this section, we consider natural distributed dynamics of debt swaps. Given an initial financial network with monotone allocation rules, our goal is to find a small number of semi-positive debt swaps to reach a local optimum, i.e., a network structure without semi-positive swaps. We denote this as the \SequenceSwapSP\ problem. While all banks (weakly) profit from a semi-positive swap, we are interested in semi-positive swaps that are $v$-improving, i.e., \emph{strictly increase the assets of a given bank $v$}. This is natural, e.g., when $v$ is a creditor that profits strictly in every such swap, or $v$ is a bank that is ``too big to fail'' and should strictly profit from network changes. We denote the problem with semi-positive $v$-improving swaps by \SequenceSwapSPv. Given a financial network, can we reach a local optimum by a sequence of polynomially many semi-positive $v$-improving swaps? 

We consider this question for the broad class of edge-ranking rules. In a given network with $m$ edges, there are at most $\genfrac(){0pt}{1}{m}{2}$ edge swaps. Recall that we can compute the clearing state for edge-ranking rules in strongly polynomial time~\cite{BertschingerStrategicPayments}. Hence, for a given network, we can also enumerate all semi-positive debt swaps in strongly polynomial time.

\subsection{Sequences of Semi-Positive Swaps}

\paragraph{Characterization} Consider a bank $v$ with edge-ranking rule. $v$ pays her outgoing edges in the order of a permutation $\pi_v$ until she runs out of funds. Assume $v$ cannot settle all debt. Let $e$ denote the unique highest-ranked outgoing edge of $v$ that is not fully paid. We call $e$ the \emph{active edge} of $v$. Clearly, all edges $e'$ with higher priority than $e$ are fully paid for, i.e., they are \emph{saturated}. Additionally, all edges $e'$ with lower priority than $e$ receive no payments. Note that the set of all active edges is \emph{acyclic} -- we can strictly increase the feasible payments along any cycle, which is impossible when the clearing state has maximal feasible payments.

Since the active outgoing edge of every bank is unique, the subgraph of active edges is simple and never contains multi-edges. When analyzing this subgraph, we use the notation for simple graphs, i.e., $e=(u,v)$ if $\deb(e)=u$ and $\cre(e)=v$, for $e\in E$. More generally, we will use the standard notation $(u,v)$ whenever it unambiguously describes one edge.

A semi-positive swap Pareto-improves the total assets of all banks. It can only increase the number of saturated edges. If a semi-positive swap saturates at least one additional edge, we call it a \emph{saturating swap}. As a second class, \emph{extension swaps} are non-saturating swaps with an interesting structural property. \changed{Recall that $v_1= \cre(e_1)$ and $v_2=\cre(e_2)$ before the trade.}

\begin{definition}[Extension Swap]
\label{def:extension-swap}
Consider a semi-positive debt swap $\sigma$ of edges $e_1$ and $e_2$ with $p_{e_1}<p_{e_2}$. The swap is called \emph{extension swap} if (1) no additional edge gets saturated and (2) in $\F$ there is an \emph{active path} $P$, i.e., a path of active edges from $v_1$ to $v_2$, (3) $l_e - p_e > p_{e_2}-p_{e_1}$ for every $e \in P$, (4) in the clearing state $p_{e_1} = \sig{p}_{e_2}$ and $p_{e_2} = \sig{p}_{e_1}$.
\end{definition}

In an extension swap, creditors $v_1$ and $v_2$ swap two edges $e_1$ and $e_2$ with different payments, where w.l.o.g.\ we assume $p_{e_2}>p_{e_1}$. The swap does not create a cycle of active edges, this would yield a new saturated edge along this cycle. Therefore, $v_1$ receives in $\sig{\F}$ additional incoming assets of $\delta_{v_1} = p_{e_2}-p_{e_1}>0$. Since there exists a path $P$ from $v_1$ to $v_2$ of active edges, each with residual liability of more than $\delta_{v_1}$, all additional assets get routed to $v_2$. The total assets of $v_2$ after the swap remain the same (see Fig.~\ref{fig:extension-swap} for an example).

\begin{figure}[t]
\begin{subfigure}[t]{0.49\textwidth}
\centering
\begin{tikzpicture}[>=stealth', shorten >=1pt, auto,
    node distance=1cm, scale=1, 
    transform shape, align=center, 
    bank/.style={circle, draw, inner sep=0.1cm}]
    \node[bank] (u1) at (0,0) {$u_1$};
    \node[bank] (u2) [right= of u1] {$u_2$};
    \node[bank] (v1) [below= of u1] {$v_1$};
    \node[bank] (v2) [right= of v1] {$v_2$};
    \node[rectangle, draw=black, inner sep=2.5pt] (xu2) [right=0.15cm of u2] {$k$};
    \draw[->] (u1) -- (v1);
    \draw[->] (u2) -- node[right, midway] {$k$} (v2);
    \draw[->] (v1) to[bend right] (v2);
\end{tikzpicture}
\end{subfigure}
\begin{subfigure}[t]{0.49\textwidth}
\centering
\begin{tikzpicture}[>=stealth', shorten >=1pt, auto,
    node distance=1cm, scale=1, 
    transform shape, align=center, 
    bank/.style={circle, draw, inner sep=0.1cm}]
    \node[bank] (u1) at (0,0) {$u_1$};
    \node[bank] (u2) [right= of u1] {$u_2$};
    \node[bank] (v1) [below= of u1] {$v_1$};
    \node[bank] (v2) [right= of v1] {$v_2$};
    \node[rectangle, draw=black, inner sep=2.5pt] (xu2) [right=0.15cm of u2] {$k$};
    \draw[->] (u1) -- (v2);
    \draw[->] (u2) -- node[right=0.1, pos=0.3] {$k$} (v1);
    \draw[->] (v1) to[bend right] node[below, midway] {$k$} (v2);
\end{tikzpicture}
\end{subfigure}
\caption{Minimal example of an extension swap where all edges have weight $M>k$ and edge labels indicate payments in the clearing state. The swap extends the path of $u_2$'s payments to reach $v_2$, increasing $v_1$'s total assets of $v_1$ along the way.}
\label{fig:extension-swap}
\end{figure}

The next proposition proves an interesting and very useful characterization.

\begin{restatable}{proposition}{propExtension}\label{prop:extension-swaps}
For a given financial network with edge-ranking rules, every semi-positive debt swap is either a saturating or an extension swap.
\end{restatable}

\begin{proof}
Consider a semi-positive swap $\sigma$ with edges $e_1,e_2$ in a financial network $\F$. W.l.o.g.\ assume $\sig{a}_{v_1} > a_{v_1}$ and $\sig{a}_{v_2} = a_{v_2}$. Suppose $\sigma$ is not a saturating swap. We show that $\sigma$ is an extension swap by proving properties (1)-(4) of the definition.

Property (1) is trivial. As a consequence, $\sigma$ does not close a new cycle of active edges; otherwise, since the set of active edges in $\sig{\F}$ is acyclic, an additional edge is saturated after the swap.

Construct the auxiliary network $\HH$ by applying source-sink equivalency to edges $e_1,e_2$ and nodes $v_1,v_2$ and simulate the debt swap by increasing the external assets of $s_{v_1}$ and $s_{v_2}$ by $\delta_{v_1}=\sig{a}_{v_1}-a_{v_1}$ and $\delta_{v_2}=0$, respectively. Due to Lemma~\ref{lem:no-pos-swap1} (3), $s_{v_1}$ routes the entire additional assets of $\delta_{v_1}$ to $\tvo,\tvt,t_{u_1}$ and $t_{u_2}$. Every bank has at most one active outgoing edge. Hence, before the swap there is exactly one path $P = \{(s_{v_1},w_1),(w_1,w_2),\ldots \}$ of active edges that ends in exactly one of the four sinks. Since no edge gets saturated, $s_{v_1}$ forwards all additional assets along $P$ to that sink. 

This sink must be either $\tvt$ or $t_{u_2}$ due to Lemma~\ref{lem:no-pos-swap2}, (2). Suppose it is $t_{u_2}$. Then $(u_2,v_2)$ is not active in $\F$, since otherwise $(P \setminus \{(s_{v_1},w_1)\}) \cup \{ (v_1,w_1), (u_2,v_1)\}$ would be a cycle of active edges in $\sig{\F}$. However, if $(u_2,v_2)$ is not active in $\F$, then neither is $(u_2,t_{u_2})$ in $\HH$. Hence, $t_{u_2}$ cannot be the sink of $P$.

As a consequence, $\tvt$ must be the sink of $P$. This implies existence of the active path and shows property (2). Since $t_{u_1}$ and $t_{u_2}$ are not involved in $P$ and there is no new cycle of active edges, these sinks receive the same assets as before the swap. This implies the same payments on swapped edges, $p_{e_1} = \sig{p}_{e_2}$ and $p_{e_2} = \sig{p}_{e_1}$. As such, property (4) holds.

Since there are no new cycles of active edges, we see that $\delta_{v_1} = \sig{p}_{e_1} - p_{e_1} = p_{e_2} - p_{e_1}$, so indeed $p_{e_1} < p_{e_2}$. Thus, $v_1$ gains additional assets of $\delta_{v_1}$ from the swapped edge. Similarly, $v_2$ looses assets of $\delta_{v_1}$ from the swapped edge. Since $\delta_{v_2} = 0$, this implies that in $\HH$ all additional assets of $\delta_{v_1}$ must get routed completely from $s_{v_1}$ along $P$ to $\tvt$. Since no additional edge gets saturated, we must have $\ell_e - p_e > p_{e_2} - p_{e_1}$, i.e., property (3) follows. As a conclusion, $\sigma$ must be an extension swap.
\end{proof}

We further classify extension swaps based on the activation status of the involved edges.
\begin{definition}[Classes of Extension Swaps]
    \label{def:active-extension-swap}
    Consider extension swap $\sigma$ of edges $e_1$ and $e_2$ with $p_{e_1} < p_{e_2}$. In a \emph{non-active} swap both edges are nonactive. In a \emph{semi-active} swap, either $e_1$ or $e_2$ is active. In an \emph{active} swap, both $e_1$ and $e_2$ are active.
\end{definition}
Observe that in a non-active swap, the edges $e_1$ and $e_2$ are each either fully saturated or receive no payments. Since $p_{e_1}<p_{e_2}$, we have both $p_{e_1}=0$ and $p_{e_2} = l_{e_2}$. In a semi-active swap, $p_{e_1} = 0$ or $p_{e_2} = l_{e_2}$. Finally, in an active swap, $0 \le p_{e_1} < p_{e_2} < l_{e_2}$.

\paragraph{Periods, Phases, Active Extension Swaps}
We now discuss \changed{the construction of} short sequences of swaps and start with saturating swaps.

\begin{observation}
    \label{obs:saturating}
    Given an initial financial network with edge-ranking rules and $m$ edges, every sequence of semi-positive debt swaps can contain at most $m$ saturating ones.
\end{observation}

We define a \emph{period} as a sequence of consecutive extension swaps. Every sequence of semi-positive swaps has at most $m+1$ periods. We turn attention to non-active and semi-active extension swaps. Every swap has exactly one non-profiting creditor. We classify non-active and semi-active swaps in a period based on this bank.
\begin{lemma}
    \label{lem:non-semi-active}
    For every $v \in V$, every single period can contain at most $2 \cdot |E^-(v)|$ many non-active or semi-active extension swaps such that $v$ is the non-profiting creditor.
\end{lemma}

\begin{proof}
    For each semi-active swap of this kind, one of two possible events occurs: Either (1) an incoming non-active saturated edge of $v$ is changed to an active (non-saturated) one, or (2) an incoming active edge of $v$ with positive payment is changed to a non-active one with 0 payment. In a non-active swap, a sequential composition of both events occurs: An incoming non-active saturated edge of $v$ is changed directly to a non-active one with 0 payments.
     
    In a semi-positive swap, the payment on every edge can only weakly increase. No edge gets saturated (this would end the period). Hence, no active edge can become non-active (neither saturated, nor with 0 payments). Events (1) and (2) cannot re-occur for an incoming edge of $v$. There are at most $2 \cdot |E^-(v)|$ events for node $v$ in the period, and at most $2 \cdot |E^-(v)|$ many non-active or semi-active extension swaps with $v$ as the non-profiting creditor.
\end{proof}

Hence, from an initial network with $m$ edges, there are at most $O(m^2)$ saturating, non-active, or semi-active extension swaps. The challenge lies in bounding the number of active extension swaps. We define a \emph{phase} as a sequence of consecutive active extension swaps. For bounding the length of swap sequences to a polynomial, Observation~\ref{obs:saturating} and Lemma~\ref{lem:non-semi-active} allow to restrict attention to the number of swaps in a single phase. 

Recall that the set of active edges is cycle-free. Every node has at most one outgoing active edge. Consequently, the set of active edges forms a forest of in-trees, where all edges are directed towards the parent node. We denote the set of these in-trees by $\mathcal{T}$. In an active extension swap, the active path $P$ runs from $v_1$ to $v_2$. As such, we can only have active extension swaps where all four banks involved in the swap are in the same tree $T \in \mathcal{T}$, and the non-profiting creditor $v_2 \in T$ is an ancestor of all $v_1, u_1, u_2 \in T$. 

We first focus on sequences of active extension swaps that strictly improve a given bank $v$. Afterwards, we consider the more general case of arbitrary active extension swaps, where we do not condition on a particular bank benefiting.

\paragraph{Improvement for a Given Bank}
We analyze sequences of semi-positive swaps that are $v$-improving for a given bank $v$, i.e., in every swap $v$ strictly profits. Note that $v$ may not be involved in the swap as a creditor bank herself. Our main result is the following.
\begin{theorem}
    \label{thm:shortSequenceV}
    For every bank $v$ in a financial network with edge-ranking rules, \emph{every} sequence of $v$-improving semi-positive swaps has polynomial length. \SequenceSwapSPv\ can be solved in polynomial time.
\end{theorem}
\begin{proof}
\begin{figure}[t]
	\centering
	%\resizebox{!}{0.15\textwidth}{
	\begin{tikzpicture}[>=stealth', shorten >=1pt, auto,
		node distance=1cm, scale=1, 
		transform shape, align=center, 
		bank/.style={circle, draw, inner sep=1},
		triangle/.style = {regular polygon, regular polygon sides=3, rotate=297, draw, inner sep=4pt},
		btriangle/.style = {regular polygon, regular polygon sides=3, rotate=0, draw, inner sep=4pt}]
		\node[circle, draw, inner sep=5pt] (v) at (0,0) {};
		\node[circle, draw, inner sep=3pt] (u1) [above right=0.1cm and 1cm of v] {$v$};
		\node[circle, draw, inner sep=5pt] (u2) [above right=0.1cm and 1cm of u1] {};
		\node[bank] (v2) [above right=0.1cm and 1cm of u2] {$v_2$};
		\node[btriangle] (tu1) [below=0.75cm of u1] {};
		\node[btriangle] (tu2) [below=0.75cm of u2] {};
		%\node[btriangle] (tu3) [above=of u1] {};
		\node[btriangle] (tv2) [below=0.75cm of v2] {};
		\node[bank] (v1) [below left =0.1cm and 1cm of v] {$v_1$};
		\node[regular polygon, regular polygon sides=3, rotate=63, draw, inner sep=4pt] (tv) [below right=0.6cm and 1cm of v1] {};
		\node[triangle] (tv1) [below left=0.6cm and 1cm of v1] {};
		\node[btriangle] (tv12) [below=0.75cm of v1] {};
		\draw[->, line width=0.04cm] (v1) to (v);
		\draw[->, line width=0.04cm] (v) to (u1);
		\draw[->, line width=0.04cm] (u1) to (u2);
		\draw[->, line width=0.04cm] (u2) to (v2);
		\draw[->] (tv1) to (v1);
		\draw[->] (tv12) to (v1);
		\draw[->] (tv) to (v1);
		\draw[->, thick, blue!50!white] (tv1) to node[below=0.1cm, midway] {\footnotesize{\textcolor{black}{$p_{e_1}$}}} (v1);
		\draw[->] (tu1) to (u1);
		\draw[->] (tu2) to (u2);
		%\draw[->] (tu3) to node[left, midway] {\footnotesize{}} (u2);
		\draw[->, thick, blue!50!white] (tv2) to node[right=0.1cm, pos=0.4] {\footnotesize{\textcolor{black}{$p_{e_2}$}}} (v2);
	\end{tikzpicture}%}
	\caption{Schematic picture of a $v$-improving semi-positive swap of active edges where payments $p_{e_2}>p_{e_1}$. Triangles indicate subtrees. When the blue edges are swapped then payments on the thick black edges strictly increase.}
	\label{fig:shortSequencev}
\end{figure}
    It suffices to prove the result for active extension swaps in a single phase. Consider the tree $T$ containing $v$. Suppose there is an active extension swap, then the set of banks that strictly profit lie on the active path $P$, with the exception of creditor $v_2$. Hence, if such a swap is also $v$-improving, $v$ is a node on the active path, the creditor $v_1$ is $v$ or a descendant of $v$, and the creditor $v_2$ is an ancestor of $v$. As such, all active edges on the path $P_v$ from $v$ to the root of $T$ cannot be swapped. In terms of payments, the edges whose payments are strictly increased by the swap are located in the subtree $T_v$ rooted at $v$ or on the path $P_v$. The swap keeps all other payments the same.

    Note that for the swap, $p_{e_1} < p_{e_2}$, i.e., the edge $e_2$ with higher payment is incoming to the ancestor $v_2$ of $v$ and gets swapped to $v_1 = v$ or a descendant $v_1$ of $v$. For $v_2$, the only incoming edge that increases in payment is the one on $P_v$. All other incoming edges of $v_2$ only keep their payments (when not being swapped) or are replaced by ones with strictly decreased payments (when being swapped). In contrast, for each edge, payments are non-decreasing over time. This implies that if edge $e_2$ is swapped away from $v_2$ for some edge $e_1$, edge $e_2$ cannot re-appear at $v_2$ via a swap with $e_1$ or any edges swapped with $e_1$.
    For a schematic illustration, see Fig.~\ref{fig:shortSequencev}.

    Hence, a single phase contains at most $m^2$ many $v$-improving active extension swaps.
\end{proof}

This result is potentially surprising, since for other allocation rules there are indeed extremely long sequences of such swaps. 
Indeed, with proportional rules there are exponentially long sequences of $v$-improving semi-positive swaps.

\begin{restatable}{proposition}{propProportionExp} \label{prop:proportionExp}
    There is a financial network with $n$ banks, $n$ edges and proportional allocation rule, from which there is a sequence of $\Omega(2^n)$ many $v$-improving semi-positive swaps.
\end{restatable}

\begin{proof}
\begin{figure}[t]
	\centering
	\resizebox{!}{0.2\textwidth}{
	\begin{tikzpicture}[>=stealth', shorten >=1pt, auto,
		node distance=1.5cm, scale=1, 
		transform shape, align=center, 
		bank/.style={circle, draw, inner sep=1}]
		\node[circle, draw, inner sep=3pt] (v) at (0,0) {$v$};
		\node[bank] (w0) [right=of v] {$w_0$};
		\node[bank] (w1) [above=of v] {$w_1$};
		\node[bank] (u0) [right=of w1] {$u_0$};
		\node[bank] (u1) [right=of u0] {$u_1$};
		\node (dots) [right=of u1] {$\cdots$};
		\node[bank] (um) [right=of dots] {$u_{m}$};
		\node[rectangle, draw=black, inner sep=2.5pt] (xu0) [above=0.1cm of u0] {\small{$2^0$}};
		\node[rectangle, draw=black, inner sep=2.5pt] (xu1) [above=0.1cm of u1] {\small{$2^1$}};
		\node[rectangle, draw=black, inner sep=2.5pt] (xum) [above=0.1cm of um] {\small{$2^{m}$}};
		\draw[->] (v) to[bend left] (w0);
		\draw[->] (w0) to[bend left] (v);
		\draw[->] (v) to[bend left] (w1);
		\draw[->] (w1) -- (u0);
		\draw[->] (u0) to node[below, midway] {\footnotesize{$2^1-1$}} (u1);
		\draw[->] (u1) to node[below, pos=0.6] {\footnotesize{$2^2-1$}} (5.5,2.05);
		\draw[->] (7,2.05) to node[below, pos=0.4] {\footnotesize{$2^{m}-1$}} (um);
	\end{tikzpicture}}
	\caption{Financial network with proportional allocation rules that admits a sequence of debt swaps with exponential length, where $m=n-4$. Each edge has weight $2^n$ and edge labels indicate payments.}
	\label{fig:expo}
\end{figure}

    We construct the initial network with a simple underlying graph as follows. There are $n$ nodes and $n$ edges. Every node has at most one incoming edge. For simplicity, all edges have a sufficiently large liability of $2^n$. There is a single node $v$ with out-degree 2, a single node with outdegree 0. All nodes have in-degree 1. The node $v$ is involved in exactly one directed cycle of two banks $v$ and $w_0$ and edges $(v,w_0)$, $(w_0,v)$. These banks have no external assets. The second outgoing edge of $v$ leads into a path with the remaining $n-2$ nodes. We denote the nodes in the path by $v,w_1,u_0,u_1,\ldots,u_{n-4}$ in the order of their occurrence. Bank $u_i$ has external assets of $2^i$, for every $i=0,\ldots,n-4$. In the initial clearing state, there are no payments on cycle edges $(v,w_0)$, $(w_0,v)$, and on edges $(v,w_1)$ $(w_1,u_0)$. Edge $(u_i,u_{i+1})$ has payment $\sum_{j=0}^i 2^j = 2^{i+1} - 1$, for every $i=0,\ldots,n-5$. A visualization can be seen in Figure~\ref{fig:expo}. 

    We outline the first four steps of an exponentially long sequence for clarity. Our sequence maintains the structural invariant that the network has a single cycle involving $v$ and a path from $v$ to $u_{n-4}$. The nodes $w_0$ and $w_1$ always remain the two out-neighbors of $v$. They alternate in being in the cycle or on the path. The assets of $v$ strictly improve by 2 in every swap.
    
    First, swap edges $(u_0,u_1)$ and $(w_0,v)$. Then the path is $v,w_0,u_1,u_2,\ldots,u_{n-4}$, and the cycle is $v, w_1, u_0$. Since $u_0$ pays its external assets of 1 in the cycle, and $v$ splits its assets evenly among $(v,w_0)$ and $(v,w_1)$, the clearing state has a payment of 2 on edge $(u_0,v)$, a payment of 1 on $(v,w_0)$, $(v,w_1)$ and $(w_1,u_1)$, and a payment of $2^{i+1}-1$ on every edge $(u_i,u_{i+1})$. Clearly, this is a $v$-improving, semi-positive swap.

    Second, swap edges $(u_0,v)$ and $(u_1,u_2)$. Then the the path is $v,w_1,u_0,u_2,\ldots,u_{n-4}$, and the cycle is $v, w_0, u_1$. Since $u_1$ pays its external assets of 2 in the cycle, the clearing state has a payment of 4 on edge $(u_1,v)$, a payment of 2 on $(v,w_0)$, $(v,w_1)$ and $(w_1,u_0)$, a payment of 3 on $(u_0,u_2)$, and $2^{i+1}-1$ on every other edge $(u_i,u_{i+1})$. Clearly, this is a $v$-improving, semi-positive swap.

    Third, swap edges $(w_0,u_1)$ and $(u_0,u_2)$. Then the path is $v,w_0,u_2,\ldots,u_{n-4}$, and cycle is $v, w_1, u_0, u_1$. The clearing state has a payment of 6 on edge $(u_1,v)$, a payment of 3 on $(v,w_0)$, $(v,w_1)$ and $(w_1,u_0)$, payment 4 on $(u_0,u_1)$. Finally, we have a payment of $2^{i+1}-1$ on every other edge $(u_i,u_{i+1})$. 

    Forth, swap edges $(u_1,v)$ and $(u_2,u_3)$. Then the path is $v,w_0,u_0,u_1,u_3,\ldots,u_{n-4}$, and the cycle is $v, w_1, u_2$. The clearing state has a payment of 8 on edge $(u_2,v)$, a payment of 4 on $(v,w_1)$, $(w_1,u_2)$ and $(v,w_0)$, $(w_0,u_0)$. There is a payment of 5 on $(u_0,u_1)$, 7 on $(u_1,u_3)$ and $2^{i+1}-1$ on every other edge $(u_i,u_{i+1})$. 

    Now at this point, we have moved $u_2$ into the cycle, and reset $u_0$ and $u_1$ to their original positions as the first $u$-nodes on the path. Thus, by repeating steps 1-3, we can bring $u_0, u_1$ in front of $u_2$ in the cycle, thus arriving at a stage with payment 14 on edge $(u_2,v)$ and 7 on both $(v,w_0)$ and $(v,w_1)$. Then, by swapping $(u_2,v)$ and $(u_3,u_4)$, we move $u_3$ in the cycle and reset $u_0$, $u_1$ and $u_2$. Thus, we can again repeat the entire sequence. In each swap, the payment on $(v,w_0)$ and $(v,w_1)$ increases by exactly 1.

    More generally, suppose we have nodes $v, w_1, u_0, u_1, \ldots, u_j$ in the cycle in that order, and we have nodes $v, w_0, u_{j+1},\ldots,u_{n-4}$ on the path in that order. At this point, we have payments of $2^{j+1}-1$ on both $(v,w_0)$ and $(v,w_1)$. Now by swapping $(u_j,v)$ and $(u_{j+1}, u_{j+2})$, we have nodes $v, w_0, u_{j+1}$ in the cycle and reset $u_0,\ldots,u_j$ on the path $v, w_1, u_0,\ldots, u_j, u_{j+2},\ldots,u_{n-4}$. Then we can repeat the entire sequence executed so far and move to $v, w_1, u_0, \ldots, u_{j+1}$ in the cycle and $v, w_0, u_{j+2},\ldots,u_{n-4}$ on the path. The last doubling step in this exponential construction occurs when $u_{n-3}$ moves into the cycle, and the sequence ends when the cycle is $v,w_1,u_0,\ldots,u_{n-5}$ and the path $v, w_0, u_{n-4}$. At this point, we have made $2^{n-4}-1$ semi-positive $v$-improving swaps.
\end{proof}

%%%%%%%%%%
\paragraph{Beyond Improvement for a Given Bank}
In the remainder of the section, we concentrate on arbitrary active extension swaps (without the property that they are $v$-improving). Bounding the length of such sequences represents a very challenging problem. We make some initial progress on this problem: For a non-trivial class of instances, we construct a \emph{single} short sequence of such swaps in each phase.

For the remainder of the section, we focus on instances with two network conditions, namely \changed{\emph{a constant number of active incoming edges} and \emph{sufficient residuals}}. First, suppose that throughout every node has at most $d$ incoming active edges, for some constant $d$. For example, $d=1$ means that every network created in the sequence has at most one incoming active edge for each $v \in V$. \changed{In particular, the number of active incoming edges is bounded by $d$ when every node has an \emph{indegree} of at most $d$.}
\begin{proposition}
    \label{prop:d=1}
    Consider an initial financial network with edge-ranking rules. Every sequence of semi-positive swaps, during which every node has at most $d=1$ active incoming edges, has polynomial length. 
\end{proposition}
\begin{proof}
  In an active extension swap, the non-profiting creditor $v_2$ has two incoming active edges. Thus, if $d=1$, there can be no active extension swap. All swaps are saturating, semi-active or non-active swaps. Their number is bounded by $O(m^2)$.
\end{proof}
When $d \ge 2$, we consider a second condition -- the network should have sufficient residuals.

\begin{definition}[Sufficient Residuals]
    In a financial network $\F$, a \emph{potential active swap} is a pair of active edges $e_1=(u_1,v_1)$ and $e_2=(u_2,v_2)$ such that $l_{e_1} = l_{e_2}$, $p_{e_1} < p_{e_2}$, and there is a path of active edges from $v_1$ to $v_2$. The network $\F$ has \emph{sufficient residuals} if for every potential active swap, $l_e - p_e > p_{e_2} - p_{e_1}$ for all $e \in P$, i.e., the potential active swap is an active extension swap.% (s versa).
\end{definition}
Intuitively, in a potential active swap all conditions for an extension swap are fulfilled, with the exception of the residual liabilities on the path $P$ of active edges. If a network has \emph{sufficient residuals}, then a potential active swap is indeed an active extension swap, since there is sufficient residual liability on $P$ to forward the marginal increase of total assets for $v_1$ towards $v_2$. For example, if the liabilities of active edges are sufficiently high compared to clearing payments, then the increases in payments will not ``exhaust'' the liabilities, and we are guaranteed to maintain the property of sufficient residuals throughout the swap sequence. 

Let us examine active extension swaps in networks with sufficient residuals.
\begin{corollary}\label{cor:lowestToRoot}
 For a financial network with edge-ranking rules and sufficient residuals, the pair $e_1$ and $e_2$ represents an active extension swap \emph{if and only if} all conditions hold: (1) $l_{e_1}=l_{e_2}$, (2) $p_{e_1} < p_{e_2}$, (3) $v_2$ is an ancestor of $v_1, u_1$ and $u_2$ in an in-tree $T$ of active edges, and (4) $u_2$ is not an ancestor of $v_1$ in $T$. 
\end{corollary}    
Clearly, if $v_2$ is an ancestor of $v_1$, there is a path $P$ of active edges, so conditions (1)-(4) characterize potential active swaps (and hence, active extension swaps). As a consequence, our goal now is to swap edges with small payments as high as possible in the tree. In particular, consider edges classified by their liability values. For each class, we intend to swap the active edges with lowest payments to the root of the in-tree. If there are active edges from the same class (a) at the root and (b) deeper in a different subtree, s.t.\ the root edge has strictly more payments, then Corollary~\ref{cor:lowestToRoot} implies there exists an active extension swap.

Formally, we construct a sequence of active extension swaps in a single phase ``top-down'' in each in-tree. In the first \emph{stage}, we only execute active swaps involving a node $v \in V_1$, where $V_1$ is the set containing the root of every $T \in \mathcal{T}$. If there are no more such active extension swaps (or there is a swap from another class and the phase can be ended), then the first stage ends. In the second stage, for every $v \in V_1$ with $v$ root of $T$, the set $V_2$ contains all direct children of $v$ in $T$. We consider active extension swaps that involve a node $v \in V_2$. The stage ends when no more such active extension swaps exist (or the phase can be ended by a swap of another class). More generally, in stage $i$ we consider active extension swaps involving nodes of depth $i-1$ in the trees.

\begin{lemma}
    \label{lem:stageInduction}
    Suppose we are given an initial financial network with edge-ranking rules, sufficient residuals and constant $d$.
    For every $v \in V_i$ there is no active extension swap involving $v$ in any stage $j > i$.
\end{lemma}

\begin{proof}
    We prove the lemma by induction. Consider stage 1 and $V_1$. Since $v \in V_1$ is the root of some in-tree $T$, it has no outgoing active edge. It can only be part of an extension swap as a non-profiting creditor. At the end of stage 1 there is no active extension swap with $v \in V_1$. Consider a liability value and suppose $v$ has $\ell$ incoming edges of this liability. Then, by Corollary~\ref{cor:lowestToRoot}, the $\ell - 1$ edges with lowest payments in this liability class in $T$ must be incident to $v$. If the edge with $\ell$-smallest payment in the class is not incident to $v$, it must be located in a different subtree than the $\ell-1$ other edges (otherwise there is an active extension swap with the remaining incoming edge of $v$ from that liability class). This property holds for each class of edges with the same liabilities.
        
    For contradiction, consider the first time at which an extension swap involving $v$ arises in a subsequent stage $j > i$. This means there is an edge incident to $v$ that can be swapped with one that is deeper in the tree and has less payment. However, this is impossible: In any previous extension swap the edges that received strictly increased payments lie on the path $P$ between the creditors. Since $v$ was not involved in any extension swap, the edges incident to $v$ still have the payments they had at the end of stage 1. Moreover, all active edges get swapped only in the respective subtrees. As such, no active extension swap with the edges of $v$ can arise. This proves the base case.

    For the hypothesis, assume the lemma holds for all stages $k = 1,\ldots,i-1$. Consider stage $i$ and a node $v \in V_i$. At the end of stage $i$, there are no extension swaps involving node $v \in V_i$. Again, consider the first time at which an extension swap involving $v$ arises in a subsequent stage $j > i$. By hypothesis, there are no extension swaps involving nodes from stages $k < i$ in this time period. Hence, in the extension swap $v$ must be the non-profiting creditor. All other nodes must be from the subtree $T_v$ rooted at $v$. We can repeat the arguments from the base case with the root and $T$ for $v$ and $T_v$. This proves the inductive step.  
\end{proof}

It remains to construct a short sequence of active extension swaps within a single stage $i$. By Lemma~\ref{lem:stageInduction}, for each $v \in V_i$, no ancestor is involved in any extension swap of stage $i$. Hence, we restrict attention to the subtree rooted in $v$. We assume $v$ is a root of a separate subtree $T$, and we must bound the number of active extensions swaps involving $v$. 

\begin{lemma}
    For any financial network with edge-ranking rules and sufficient residuals, if every $v \in V_i$ has at most constantly many incoming active edges, then there can be at most a polynomial number of active extension swaps in stage $i$.
\end{lemma}

\begin{proof}
    Consider $v \in V_i$ with $d_v \leq d$ incoming active edges, for some constant $d$. Active extension swaps always keep this number the same for every node. Since $v$ is a non-profiting creditor in every swap of stage $i$, the total assets of $v$ do not change during the stage. However, in each swap, there is at least one incoming edge of $v$ (on the active path of the swap) that strictly increases in payments. Since payments on each edge do not decrease, the same subset of $d_v$ incoming active edges cannot repeat during the stage. As a consequence, at the end of each swap, $v$ has a different subset of neighbors for its incoming active edges. The number of different sets of neighbors is bounded by $\genfrac(){0pt}{1}{n}{d_v}$, which is polynomial for constant $d_v$.
\end{proof}

Together these insights prove Theorem~\ref{thm:shortSequence}, which summarizes our result.

\begin{theorem}
    \label{thm:shortSequence}
    The \SequenceSwapSP\ problem for financial networks with edge-ranking rules can be solved in polynomial time if, for all financial networks encountered during the sequence, the network has sufficient residuals and every node has at most a constant number of active incoming edges.
\end{theorem}
\begin{openproblem*}
    \label{op:shortSequence}
    Can the \SequenceSwapSP\ problem for financial networks with edge-ranking rules be solved in polynomial time?
\end{openproblem*}

\subsection{Sequences of Arbitrary Swaps}
\label{sec:localArb}

\changed{
\begin{figure}[t]
\centering
\begin{tikzpicture}[>=stealth', shorten >=1pt, auto,
    node distance=1cm, %scale=0.8, 
    transform shape, align=center, 
    bank/.style={circle, draw, inner sep = 1}]
    \node[bank] (v1) at (0,0) {$v_1$};
    \node[bank] (v2) at (1.5,0) {$v_2$};
    \node[bank] (v3) at (3,0) {$v_3$};
    \node[bank] (v4) at (4.5,0) {$v_4$};
    \node[bank] (v5) at (6,0) {$v_5$};
    \draw[->] (0,0.7) to node[right] {2} (v1);
    \draw[->] (v1) to node[left=0.1cm,pos=0.4]{1} (-0.25,-0.7);
    \draw[->] (v1) to node[right=0.1cm,pos=0.4]{2} (0.25,-0.7);
    \draw[->] (1.5,0.7) to node[right] {2} (v2);
    \draw[->] (v2) to node[left=0.1cm,pos=0.4]{1} (1.25,-0.7);
    \draw[->] (v2) to node[right=0.1cm,pos=0.4]{2} (1.75,-0.7);
    \draw[->] (3,0.7) to node[right] {2} (v3);
    \draw[->] (v3) to node[right] {2} (3,-0.7);
    \draw[->] (4.5,0.7) to node[right] {1} (v4);
    \draw[->] (v4) to node[right] {1} (4.5,-0.7);
    \draw[->] (5.75,0.7) to node[left=0.1cm, pos=0.4] {1} (v5);
    \draw[->] (6.25,0.7) to node[right=0.1cm, pos=0.4] {1} (v5);
\end{tikzpicture}\\
\vspace{1cm}
\begin{tikzpicture}[>=stealth', shorten >=1pt, auto,
    node distance=1cm, %scale=0.8, 
    transform shape, align=center, 
    bank/.style={circle, draw, inner sep = 1}]
    \def\r{1}
    \node[bank] (v1) at (210:\r) {$v_1$};
    \node[bank] (v2) at (330:\r) {$v_2$};
    \node[bank] (v3) at (90:\r) {$v_3$};
    \node[bank] (v4) [below = of v1] {$v_4$};
    \node[bank] (v5) [below = of v2] {$v_5$};
    \draw[->] (v2) to node[above] {2} (v1);
    \draw[->] (v3) to node[right=0.1cm, pos=0.4] {2} (v2);
    \draw[->] (v1) to node[left=0.1cm, pos=0.6] {2} (v3);
    \draw[->] (v1) to node[below=0.1cm,pos=0.1] {1} (v5);
    \draw[->] (v4) to node[below]{1} (v5);
    \draw[->] (v2) to node[below=0.1cm,pos=0.1]{1} (v4);
\end{tikzpicture}
\hspace{1.5cm}
\begin{tikzpicture}[>=stealth', shorten >=1pt, auto,
    node distance=1cm, %scale=0.8, 
    transform shape, align=center, 
    bank/.style={circle, draw, inner sep = 1}]
    \def\r{1}
    \node[bank] (v1) at (210:\r) {$v_1$};
    \node[bank] (v2) at (330:\r) {$v_2$};
    \node[bank] (v3) at (90:\r) {$v_3$};
    \node[bank] (v4) [below = of v1] {$v_4$};
    \node[bank] (v5) [below = of v2] {$v_5$};
    \draw[->] (v1) to node[above] {2} (v2);
    \draw[->] (v2) to node[right=0.1cm, pos=0.6] {2} (v3);
    \draw[->] (v3) to node[left=0.1cm, pos=0.4] {2} (v1);
    \draw[->] (v1) to node[left] {1} (v4);
    \draw[->] (v4) to node[below]{1} (v5);
    \draw[->] (v2) to node[right]{1} (v5);
\end{tikzpicture}
\caption{\changed{Example of a set of nodes together with corresponding stubs on the top, and two graphs with consistent sets of edges below.}}
\label{fig:stubs}
\end{figure}
}

In this section, we consider sequences of arbitrary swaps. In contrast to semi-positive ones, these swaps have no natural criterion of progress. We again focus on swaps in which a single bank $v$ has to improve strictly. Unfortunately, sequences of such arbitrary $v$-improving swaps can be very long, for every monotone allocation rule (including edge-ranking and proportional rules). Our results show that there are initial networks from which \emph{every} sequence of $v$-improving swaps to a local optimum (in which no such swaps exist) is exponentially long.

Indeed, we obtain a more general statement. Consider a  maximization problem for the total assets of a given bank $v$ defined as follows. In an instance of the \MaxAssetsv\ problem, we are given a set of banks $V$. For each bank $v \in V$ we are given a non-negative value for the external assets as well as a number of incoming and outgoing \emph{half-edges} or \emph{stubs}. For each such stub $i$ of $v$, we have a positive weight $l_i > 0$. Clearly, these stubs shall be \emph{consistent} with at least one set of edges $E$ among $V$, i.e., there must be a feasible underlying multi-graph $G=(V,E)$ in the sense that the stubs can be matched to form $G$. In particular, each incoming stub of weight $l_i$ of each node $v$ can be matched to some outgoing stub of weight $l_i$ of some node $u \neq v$ to form an edge $e\in E$ with $\deb(e)=u$ and $\cre(e)=v$. \changed{An example is depicted in Figure~\ref{fig:stubs}.} Consistency of stubs with \emph{at least one} set $E$ can be tested in polynomial time by solving a straightforward bipartite matching problem. For the remainder, we assume that it holds.

An instance of \MaxAssetsv\ has an edge-ranking rule $\vecf$ that specifies the order in which banks pay their outgoing stubs. The rule is given implicitly by a \emph{clearing oracle} $\mathcal{O}$. A clearing oracle for rule $\vecf$ receives as input the remaining financial network $(G, \deb, \cre, \vecl, \veca^x)$ and outputs the clearing state. Note that a clearing oracle $\mathcal{O}$ can be computed in polynomial time for edge-ranking rules~\cite{BertschingerStrategicPayments}.

Observe that any set of edges $E$ that is consistent with the stubs gives rise to a financial network and a clearing state (where the latter can be computed using $\mathcal{O}$). The goal is to find such a consistent set to maximize the total assets of a given bank $v$ in the resulting clearing state. The objective function in the instance are the total assets for a given bank $v$.

Finally, for edge swaps we define a neighborhood over the set of financial networks whose edges are consistent with the stubs in the input. For a given financial network $\F$, the \Swap\ neighborhood $\mathcal{N}(\F)$ contains all networks $\F'$ that can be obtained by performing an arbitrary debt swap in $\F$. A debt swap exchanges the matching of two pairs of stubs with pairwise disjoint incident nodes and with the same weights and thereby maintains consistency.

\begin{proposition}
\MaxAssetsv\ with \Swap\ neighborhood is in the complexity class \classPLS\ whenever the clearing oracle can be computed in polynomial time.
\end{proposition}
Clearly, as noted above, we can obtain an initial network by computing a bipartite matching (indeed, \emph{every} network consistent with the stubs represents a bipartite matching of the stubs and vice versa). The \Swap\ neighborhood of a network contains $O(m^2)$ other networks. For each of these networks, the corresponding clearing state can be determined using oracle $\mathcal{O}$. Hence, the \Swap\ neighborhood $\N(\F)$ for a given network $\F$ can be searched efficiently.

When locally maximizing \MaxAssetsv\ we do not necessarily restrict attention to a sequence of $v$-improving swaps from a given initial network to a local optimum. Instead, when given an initial network, we obtain an instance of \MaxAssetsv\ by breaking all edges into stubs. The goal is to find \emph{any} local optimum network consistent with the stubs derived from the initial network. We show that this problem is \classPLS-complete \changed{\cite{johnson1988easy}, thus, unlikely to be solvable in polynomial time. It is easy to see that our \classPLS-reduction in the subsequent theorem is \emph{tight}~\cite[Definition 3]{Yannakakis03}, i.e., it preserves solutions, neighborhoods, and objective values from the \maxSat\ problem with \textit{flip} neighborhood. This implies the existence of instances and starting states that require exponentially long sequences of swaps as well as \classPSPACE-completeness of reachability questions (see also Corollary~\ref{cor:PSPACE} below).}

%For a proof of the following result, see Appendix~\ref{app:proofs}.

\begin{restatable}{theorem}{thmPLS} \label{thm:PLS}
    There is a tight \classPLS-reduction from \maxSat\ with \emph{flip} neighborhood to \MaxAssetsv\ with \Swap\ neighborhood.
\end{restatable}

\begin{proof}
\begin{figure}
\centering
\resizebox{!}{0.4\textwidth}{
	\begin{tikzpicture}[>=stealth', shorten >=1pt, auto,
		node distance=1cm, scale=1, 
		transform shape, align=center, 
		bank/.style={circle, draw, inner sep=3.1pt}]
        \node[bank] (x1) at (0,0) {$x_1$};
        \node[bank] (s1) [right=of x1] {$s_1$};
        \node[bank] (x2) [right=of s1] {$x_2$};
        \node[bank] (s2) [right=of x2] {$s_2$};
        \node (udots) [right=of s2] {$\dots$};
        \node[bank] (xk) [right=of udots] {$x_k$};
        \node[bank] (sk) [right=of xk] {$s_k$};
        \node[circle, draw, inner sep=1pt] (x1t) [below=of x1] {$x_1^T$};
        \node[circle, draw, inner sep=1pt] (x1f) [right=of x1t] {$x_1^F$};
        \node[circle, draw, inner sep=1pt] (x2t) [right=of x1f] {$x_2^T$};
        \node[circle, draw, inner sep=1pt] (x2f) [right=of x2t] {$x_2^F$};
        \node (mdots) [right=of x2f] {$\cdots$};
        \node[circle, draw, inner sep=1pt] (xkt) [right=of mdots] {$x_k^T$};
        \node[circle, draw, inner sep=1pt] (xkf) [right=of xkt] {$x_k^F$};
        \node (anchor) [below=2cm of x2f] {};
        \node[bank] (k2) [left=0.75cm of anchor] {$\kappa_2$};
        \node[bank] (k1) [left=1.5cm of k2] {$\kappa_1$};
        \node (ldots) [right=0.75cm of anchor] {$\cdots$};
        \node[bank] (kl) [right=1.5cm of ldots] {$\kappa_{\ell}$};
        \node[circle, draw, inner sep=3.5pt] (v) [below=2cm of anchor] {$v$};
        \node[rectangle, draw=black, inner sep=2.5pt] (ex1) [above=0.1cm of x1] {\small{$kM_1$}};
        \node[rectangle, draw=black, inner sep=2.5pt] (ex2) [above=0.1cm of x2] {\small{$kM_2$}};
        \node[rectangle, draw=black, inner sep=2.5pt] (exk) [above=0.1cm of xk] {\small{$kM_k$}};
        \draw[->] (x1) to node[left, midway] {\footnotesize{$kM_1$}} (x1t);
        \draw[->] (s1) to node[left, midway] {\footnotesize{$kM_1$}} (x1f);
        \draw[->] (x2) to node[left, midway] {\footnotesize{$kM_2$}} (x2t);
        \draw[->] (s2) to node[left, midway] {\footnotesize{$kM_2$}} (x2f);
        \draw[->] (xk) to node[left, midway] {\footnotesize{$kM_k$}} (xkt);
        \draw[->] (sk) to node[left, midway] {\footnotesize{$kM_k$}} (xkf);
        \draw[->] (x1t) to node[left=0.1cm, pos=0.2] {\footnotesize{$M_1^T$}} (k1);
        \draw[->] (x1f) to node[left=0.1cm, pos=0.2] {\footnotesize{$M_1^F$}} (k2);
        %\draw[->] (x2t) to node[left, pos=0.2] {\footnotesize{$M_2^T$}} (k2);
        \draw[->] (x2f) to node[right, pos=0.2] {\footnotesize{$M_2^F$}} (k1);
        \draw[->] (x2f) to node[right=0.1cm, pos=0.2] {\footnotesize{$M_2^F$}} (kl);
        \draw[->] (xkt) to node[left=0.5cm, pos=0.2] {\footnotesize{$M_2^T$}} (kl);
        \draw[->] (xkf) to node[right=0.2cm, pos=0.2] {\footnotesize{$M_2^F$}} (k2);
        \draw[->] (k1) to node[left=0.2cm, pos=0.2] {\footnotesize{$w_{\kappa_1}$}} (v);
        \draw[->] (k2) to node[right=0.2cm, pos=0.2] {\footnotesize{$w_{\kappa_2}$}} (v);
        \draw[->] (kl) to node[right=0.2cm, pos=0.2] {\footnotesize{$w_{\kappa_{\ell}}$}} (v);
    \end{tikzpicture}}
\caption{Example construction for reduction with \maxSat\ where edge labels indicate liabilities.}
\label{fig:Max2Sat}
\end{figure}

\changed{An instance} $\varphi$ of \maxSat\ with the \textit{flip} neighborhood \changed{consists of formula} $\varphi$ consists of $k$ variables $x_1,x_2\dots,x_k$ and $\ell$ clauses $\kappa_1,\ldots,\kappa_{\ell}$ where each clause is assigned a weight $w_{\kappa_j}$ \changed{and consists of at most two variables. The objective is to derive an assignment of variables that maximizes the sum of weights of satisfied clauses. In the flip neighborhood, two assignments are neighboring if they vary in the assignment of exactly one variable.} We construct an input instance to \MaxAssetsv\ as follows. First, include banks $x_i, x_i^T$ and $x_i^F$ for every variable and a bank $\kappa_j$ for every clause. Additionally add bank $v$ and auxiliary bank $s_i$ for all $i\in[k]$.

Before specifying the stubs and their weights for every bank, we describe a binary encoding for every bank corresponding to a literal. Choose integer $\ell$ such that $2^{\ell}\geq w_{\kappa_j}$ for all $j$ and $\ell > 1+\log k$. Then, set the $\ell$-th lowest bit to 1. Use the bits $2$ to $1+\log k$ to encode the index $i$ of the literal and set the first bit from the right to $0$ if the literal is $\neg x_i$ and to $1$ otherwise. All other bits are set to $0$. Observe that the weights assigned to each literal are pairwise different and at least $w_{\kappa_j}$ for all $j$. \changed{The use of these different weights limits the number of possible matchings of the stubs. In particular, when assigning one of the unique weights to only two stubs, the existence of a specific edge in each set $E$ can be enforced. We discuss the exact effects of the weights below.} For simplicity, denote the resulting value for each bank $x_i^T$ by $M_i^T$, and for each $x_i^F$ by $M_i^F$. 

We construct the stubs as follows. Bank $v$ has outdegree zero and $\ell$ incoming stubs such that there exists exactly one with weight $w_{\kappa_j}$. Each bank $\kappa_j$ has exactly one outgoing stub with weight $w_{\kappa_j}$. Furthermore, if literal $x_i$ appears in the clause, then we add an incoming stub with weight $M_i^T$ -- or a stub with weight $M_i^F$ if $\neg x_i$ appears. For every bank $x_i^T$, there is a stub for each clause the literal $x_i$ appears in, each with weight $M_i^T$. The analogous condition holds for $x_i^F$. Using $M_i = \max\{M_i^T,M_i^F\}$, we create a single incoming stub with weight $kM_i$ for each $x_i^T$ and $x_i^F$. Lastly, for every bank $x_i$ and $s_i$, we set the indegree to zero and create one outgoing stub with weight $kM_i$. Observe that no set of edges $E$ consistent with the stubs contains multi-edges.

Finally, bank $x_i$ has external assets of $kM_i$, for every $i \in [k]$.

Overall, the construction of the \MaxAssetsv\ instance is polynomial in the size of the instance of \maxSat. For an example construction see Figure~\ref{fig:Max2Sat}.

For the clearing oracle, each bank with at most one outgoing stub has a trivial clearing payment -- simply assign all assets to the outgoing stub (if any) until it is fully paid. The banks $x_i^T$ and $x_i^F$ can have multiple outgoing stubs. We argue below that in all consistent networks each of these banks has either 0 assets or enough assets to fully pay all outgoing stubs. This implies that for a given consistent network and a given bank, the payments of the bank are the same, for \emph{every} allocation rule (edge-ranking, proportional, etc., even for non-monotone ones).

It is straightforward to see whether a bank has 0 assets or can pay all edges fully. Thus, defining a clearing payment is possible in polynomial time for every bank in every consistent network, for every (or, more precisely, independent of the specific) allocation rule.

We now prove correctness of the reduction. Observe that, by construction, every consistent set of edges must contain edge $(\kappa_j,v)$ since for each $w_{\kappa_j}$ there exists exactly one bank with an outgoing and one bank with an incoming stub. Similarly, $x_i^T$ is the only bank that has outgoing stubs with weight $M_i^T$. Additionally, for every bank $\kappa_j$ with an incoming stub with weight $M_i^T$ the corresponding clause includes literal $x_i$. Hence, there is an edge $(x_i^T,\kappa_j)$ in every consistent set of edges if and only if $x_i$ appears in clause $\kappa_j$. The same argument extends to literals $x_i^F$. Thus, the consistent sets of edges only vary in the edges of banks $x_i$ and $s_i$ to $x_i^T$ and $x_i^F$. Note that for every $i$ it is possible to either choose (1) edges $(x_i,x_i^T)$ and $(s_i,x_i^F)$ or (2) $(x_i,x_i^F)$ and $(s_i,x_i^T)$, each with the same edge weight $kM_i$. In case (1) we interpret literal $x_i$ to be assigned with value TRUE and $\neg x_i$ with value FALSE. On the other hand, in case (2) literal $\neg x_i$ is interpreted as TRUE and $x_i$ as FALSE. In conclusion, every consistent set of edges corresponds to a truth assignment of variables in $\varphi$ and vice versa.

By construction, exactly one of the banks $x_i^T$ or $x_i^F$ has strictly positive incoming assets. This bank defines the corresponding assignment for the variable. By construction, this bank is then able to settle all debt (while the other one settles no debt at all). As a consequence, bank $\kappa_j$ has a fully saturated incoming edge if and only if there exists a satisfying literal in the corresponding clause. Then, by construction, $\kappa_j$ can settle all debt towards $v$. As a result, for any consistent network, the incoming assets of $v$ equal the sum of weights of satisfied clauses in $\varphi$ for the corresponding assignment. In other words, the local optima for the given instance of \maxSat\ correspond exactly to the local optimal networks w.r.t.\ $v$-improving debt swaps.
\end{proof}

This shows that \MaxAssetsv\ with \Swap\ neighborhood is \classPLS-complete. Moreover, tightness of the reduction implies that there exist initial states from which every execution of the standard local search algorithm requires exponentially many steps to reach a local optimum~\cite{Yannakakis03,schaffer1991simplelocalsearchproblems}. The standard local search algorithm for \MaxAssetsv\ yields a sequence of debt swaps.

\begin{corollary}
    \label{cor:PLS}
    It is \classPLS-complete to compute a local optimum for \MaxAssetsv\ with \Swap\ neighborhood. There is an instance of \MaxAssetsv\ and an initial network $\F$ such that every sequence of arbitrary $v$-improving swaps from $\F$ to a local optimum has exponential length.
\end{corollary}
In the reduction of Theorem~\ref{thm:PLS}, each possible network has the property that a bank has either 0 assets or enough assets to fully pay all outgoing stubs. 
The clearing oracle in these networks is trivial, and all payments are completely independent of the allocation rule. 
\begin{corollary}
    The results of Corollary~\ref{cor:PLS} generalize to \MaxAssetsv\ with the \Swap\ neighborhood for any monotone allocation rule with efficient clearing oracle.
\end{corollary}

%%%%%%%%%%%%%%%%%%%%%%%%%%%%%

\section{Optimization} \label{sec:opt}

In this section, we study (global) optimization problems based on sequences of debt swaps. In particular, for a given initial network, we are interested in constructing sequences of debt swaps to a reachable network that maximize either the resulting total assets of a specific bank or the sum of total assets for all banks. We consider sequences of semi-positive swaps, $v$-improving swaps, arbitrary swaps, and also a restriction where a specific bank $v$ must be involved in every swap.

Before considering sequences of swaps, we prove a slightly more general result. Therefore, consider the \MaxAssetsv\ problem defined in Section~\ref{sec:localArb} without the neighborhood relation. Instead of local optima we are now interested in \emph{global} ones.
In the same way, we define a \MaxAssets\ problem, where instead of the assets of $v$ the objective function is the sum of total assets of all banks. 

We again assume to have access to a clearing oracle as discussed in Section~\ref{sec:localArb}. Similar to the property in the proof of Theorem~\ref{thm:PLS}, our reductions always ensure that every reachable network has a clearing state that is independent of the allocation rule and can be computed in polynomial time. As such, all results in this section apply beyond edge-ranking rules to all monotone allocation rules.
Additionally, none of the results in this section relies on multi-edges. For this reason, we adopt the notation for simple graphs.

\begin{theorem}\label{them:opt-setCover}
    \MaxAssets\ and \MaxAssetsv\ are \classNP-hard.
\end{theorem}
\begin{proof}
\begin{figure}[t]
	\centering
	\resizebox{!}{0.35\textwidth}{
	\begin{tikzpicture}[>=stealth', shorten >=1pt, auto,
		node distance=1cm, scale=1, 
		transform shape, align=center, 
		bank/.style={circle, draw, inner sep=2.5pt}]
		\node[bank] (u1) at (0,0) {$u_1$};
		\node[bank] (u2) [above=of u1] {$u_2$};
		\node (udots) [above=0.5cm of u2] {$\vdots$};
		\node[bank] (ul) [above=0.5cm of udots] {$u_{\ell}$};
		\node[circle, draw, inner sep=1.7] (S1) at (1.5,0) {$S_1$};
	 	\node[circle, draw, inner sep=1.7] (S2) [above=of S1] {$S_2$};
		\node (Sdots) [above=0.5cm of S2] {$\vdots$};
		\node[circle, draw, inner sep=1.7] (Sl) [above=0.5cm of Sdots] {$S_{\ell}$};
		\node[bank] (x1) at (5,-1) {$x_1$};
		\node[bank] (x2) [above=of x1] {$x_2$};
		\node[bank] (x3) [above=of x2] {$x_3$};
		\node (xdots) [above=0.5cm of x3] {$\vdots$};
		\node[bank] (xk) [above=0.5cm of xdots] {$x_k$};
		\node[circle, draw, inner sep=4] (v) [right=7cm of u2] {$v$};
		\node (ancher) [right=of v] {};
		\node (wdots) [above=0.25cm of ancher] {$\vdots$};
		\node[bank] (w2) [below=0.25cm of ancher] {$w_2$};
		\node[bank] (w1) [below=of w2] {$w_1$};
		\node[bank] (wl) [above=of wdots] {$w_{\lambda}$};
		\node[rectangle, draw=black, inner sep=2.5pt] (xu1) [left=0.1cm of u1] {\small{$M$}};
		\node[rectangle, draw=black, inner sep=2.5pt] (xu2) [left=0.1cm of u2] {\small{$M$}};
		\node[rectangle, draw=black, inner sep=2.5pt] (xul) [left=0.1cm of ul] {\small{$M$}};
		\draw[->] (u1) to node[below, midway] {\footnotesize{$M$}} (S1);
		\draw[->] (u2) to node[below, midway] {\footnotesize{$M$}} (S2);
		\draw[->] (ul) to node[below, midway] {\footnotesize{$M$}} (Sl);
		\draw[->] (S1) to node[below, pos=0.2] {\footnotesize{$M_1$}} (x1);
		\draw[->] (S1) to node[below=0.1cm, pos=0.2] {\footnotesize{$M_1$}} (x3);
		\draw[->] (S2) to node[above=0.15cm, pos=0.2] {\footnotesize{$M_2$}} (x1);
		\draw[->] (S2) to node[below=0.15cm, pos=0.2] {\footnotesize{$M_2$}} (xk);
		\draw[->] (Sl) to node[below=0.2cm, pos=0.1] {\footnotesize{$M_{\ell}$}} (x2);
		\draw[->] (Sl) to node[above=0.1cm, pos=0.2] {\footnotesize{$M_{\ell}$}} (x3);
		\draw[->] (Sl) to node[above=0.1cm, pos=0.2] {\footnotesize{$M_{\ell}$}} (xk);
		\draw[->] (x1) to node[below=0.2cm, pos=0.2] {\footnotesize{$1$}} (v);
		\draw[->] (x2) to node[below=0.1cm, pos=0.2] {\footnotesize{$1$}} (v);
		\draw[->] (x3) to node[below=0.1cm, pos=0.2] {\footnotesize{$1$}} (v);
		\draw[->] (xk) to node[below=0.2cm, pos=0.2] {\footnotesize{$1$}} (v);
		\draw[->] (w1) to node[below=0.3cm, pos=0.3] {\footnotesize{$M$}} (v);
		\draw[->] (w2) to node[below=0.2cm, pos=0.3] {\footnotesize{$M$}} (v);
		\draw[->] (wl) to node[below=0.2cm, pos=0.3] {\footnotesize{$M$}} (v);
	\end{tikzpicture}}
\caption{Example construction for reduction with \setCover\ where $\lambda=\ell-c$. Edge labels indicate liabilities.}
\label{fig:setCover}
\end{figure}
    We first show the statement for \MaxAssetsv. \changed{While we eventually focus on the \vertexCover\ problem in Corollaries~\ref{cor:maxAssetHardRank} and \ref{cor:maxAssetHard} below, we will describe a reduction from the more general \setCover\ problem to avoid confusion in the reference to graph-theoretic concepts.} An instance of the \setCover\ problem is given by items $\mathcal{X}=\{x_1,x_2,\dots,x_k\}$ and sets $S_1,S_2,\dots,S_{\ell}$ such that $\bigcup_{j \in [\ell]}S_j = \mathcal{X}$, where $[\ell]:= \{1,2,\dots,\ell\}$, \changed{as well as an integer $c > 0$. The goal is to decide whether or not there exists a set cover, i.e., a subset $\mathcal{C} \subset [\ell]$ with $\bigcup_{j \in \mathcal{C}} S_j = \mathcal{X}$, that has cardinality $|\mathcal{C}| = c$.} We construct an instance of \MaxAssetsv\ with banks $v$, $x_i$ for every item and $S_j$ for every set. Moreover, we add banks $u_1,u_2,\dots,u_{\ell}$ and $w_1,w_2,\dots,w_{\ell-c}$ for some integer $c$ satisfying $0<c<\ell$. 
 
    Similar to the proof of Theorem~\ref{thm:PLS} we define a value for every set $S_j$ in binary encoding \changed{to enforce a subset of edges in every set $E$}. First, let $d$ be the smallest integer satisfying $2^{d} \geq \max_j |S_j|$ and $d > \log \ell$. Then, set the $d$-th lowest bit to $1$ and encode index $j$ using the lowest $\log \ell$ bits. All other bits are set to $0$. For simplicity, denote the resulting value by $M_j$ for every $j$ and let $M=\max_j M_j$. Observe that all values $M_j$ are pairwise different and lower bounded by $\max_j |S_j|$. The number of bits is polynomial in the input size.
    
    Banks $u_j$ and $w_j$ each have no incoming but one outgoing stub with weight $M$ where the external assets of $u_j$ are set to the same value. If item $x_i$ is element of set $S_j$ in the input instance add an incoming stub to bank $x_i$ and an outgoing stub to bank $S_j$ each with weight $M_j$. Moreover, every bank $S_i$ has one incoming stub with weight $M$. Additionally, include an outgoing stub with weight 1 to each $x_i$ and $k$ incoming stubs with weight $1$ to $v$. Finally, add $\ell-c$ incoming stubs to $v$ with weight $M$. 
    
    Banks $S_j$ are the only banks with more than one outgoing stub. We will argue below that every $S_j$ either has no assets or enough to settle all debt. For this reason, the clearing state can be computed efficiently for every monotone allocation rule.
    
    Clearly, the construction can be performed in polynomial time with respect to the input size of the \setCover\ instance. Figure~\ref{fig:setCover} illustrates an example construction.
    
    By construction, every consistent set of edges must contain $(S_j,x_i)$ with weight $M_i$ if element $x_i$ is included in set $S_j$. Moreover, since only banks $x_i$ have an outgoing stub with weight 1 and $v$ has $k$ incoming stubs with weight 1, edges $(x_i,v)$ must be contained in every set of edges. Finally, banks $w_i$ and $u_i$ each must have an edge either towards $v$ or a bank $S_j$. If bank $S_j$ has an incoming edge from a bank $w_i$ it has total assets of zero, or otherwise enough assets to pay off all debt if it has an incoming edge from $u_i$.
    
    We claim that there exists a set cover with size $c$ if and only if there exists a consistent set of edges such that $v$ has total assets of $M(\ell-c)+k$ in the corresponding network.
    
    Suppose there exists a set cover $\mathcal{C}\subset [\ell]$ with $|\mathcal{C}|= c$ for the input instance. Then, for every $j \in \mathcal{C}$ include an edge from $u_j$ to $S_j$. As a result, $S_j$ has incoming payments of $M$. For the remaining $\ell-c$ banks $S_j$, choose an incoming edge from a bank $w_i$ each. Hence, no bank $S_j$ with $j \notin \mathcal{C}$ has positive incoming payments. Observe that there are exactly $\ell-c$ many remaining banks $u_i$. For every one of them include the edge $(u_i,v)$ yielding incoming assets of at least $M(\ell-c)$ for $v$ in the resulting network. Observe that the set of banks $S_j$ with incoming assets corresponds to the collection $\mathcal{C}$ covering all elements. Hence, by construction of $M_j$ every bank $x_i$ receives incoming assets of at least 1 and forwards the payments to $v$. In conclusion, $v$ has total assets of $M(\ell-c)+k$.
    
    For the other direction, assume there exists a network with consistent set of edged where $v$ has total assets of $M(\ell-c)+k$. Because the weights of edges $(x_i,v)$ are restricted to 1 and for sufficiently large $M$, bank $v$ can only achieve incoming assets of at least $M(\ell-c)$ if $\ell-c$ banks $u_i$ pay all their assets directly to $v$. The remaining assets of $k$ then must be payments from banks $x_i$. Hence, every $x_i$ receives payments of at least 1. Consequently, there exists a collection of banks $S_j$ with cardinality $c$ covering all $x_i$.  
    
    To show the statement for $\MaxAssets$, add $B$ banks $b_1,b_2,\dots,b_B$ 
    where $b_i$ has an incoming stub with weight $B+i$ and an outgoing stub with weight $B+i+1$, for $i<B$. Add an outgoing stub with weight $B+1$ to $v$ and incoming stub with weight $2B$ to $b_B$. 
    Observe that \changed{for a sufficiently large} $B$ there exist exactly two stubs with weight $B+i$. Hence, there exists only one consistent set of edges that joins the stubs to form a path from $v$ to $b_B$. Every bank $b_i$ on the path has the same total assets as $v$. When $B$ is sufficiently large, say, $B = 3(k+\ell)$, then maximizing the sum of total assets in the network reduces to maximizing the total assets of $v$. 
\end{proof}

The reduction directly implies the following hardness of approximation.

\begin{corollary}
    \label{cor:maxAssetHardRank}
    \MaxAssets\ and \MaxAssetsv\ are \classAPX-hard.
\end{corollary}
\begin{proof}
    Recall that \vertexCover\ is \classAPX-hard in 3-regular graphs~\cite{AlimontiK00}. Interpreting a 3-regular graph with $n'$ nodes as an instance of set cover, we obtain $\ell = n'$ sets and $k = 3n'$ elements. Clearly, every feasible set cover in these instances has size $\Theta(n')$. Applying the reduction from Theorem~\ref{them:opt-setCover}, we observe that a set cover of size $c$ corresponds to bank $v$ obtaining assets $M(\ell - c) + k$, where $M > \ell$. As such, in these instances the assets are in $M \cdot \Theta(n')$ for every consistent network corresponding to a feasible set cover. As a consequence, the constant-factor approximation gap of \vertexCover\ translates to a constant-factor gap for optimizing the assets of $v$. Moreover, since the adaptation to maximizing the sum of all assets requires only $\Theta(k + \ell) = \Theta(|V|)$ additional banks, the result extends to this case as well.
\end{proof}

We proceed to show how this result implies hardness results for sequences of swaps. In the instances constructed in the \changed{proof of Theorem~\ref{them:opt-setCover}}, set $S_j$ is included in the cover $\mathcal{C}$ if and only if bank $S_j$ has positive incoming assets. Every network with consistent set of edges corresponds to some collection of sets $\mathcal{C}$, and vice versa. Suppose we are given such a network for a collection $\mathcal{C}$. We can perform a debt swap where the debtor of one bank $S_j$ is swapped from a bank $w_h$ to a bank $u_i$ or the other way around. In this way, we remove/add $S_j$ from/to the corresponding collection of subsets. Thus, we can reach the network corresponding to the optimal set cover $\mathcal{C}^*$ via a sequence of at most $|(\mathcal{C} \cup \mathcal{C}^*) \setminus (\mathcal{C} \cap \mathcal{C}^*)|$ many (arbitrary, not necessarily $v$-improving) debt swaps.

\begin{corollary}
    \label{cor:maxAssetHard}
     For a given financial network with monotone allocation rules, an efficient clearing oracle, and a bank $v$, it is \classAPX-hard to compute a sequence of arbitrary debt swaps such that (1) the total assets of $v$ are maximal or (2) the sum of total assets is maximal. 
\end{corollary}

We strengthen this result as follows. Suppose we start with the collection $[\ell]$, i.e., all sets are in the collection. The optimal set cover $\mathcal{C}^*$ can be obtained by iterative removal of sets from $[\ell]$. Each removal represents a debt swap in the corresponding financial network. It is straightforward to verify that each such debt swap is $v$-improving. Hence, from the initial network corresponding to $[\ell]$, there is a sequence of at most $\ell$ $v$-improving debt swaps to the network corresponding to $\mathcal{C}^*$.

\begin{corollary}
    \label{cor:maxAssetImproveHard}
     For a given financial network with monotone allocation rules, an efficient clearing oracle, and a bank $v$, it is \classAPX-hard to compute a sequence of $v$-improving debt swaps such that (1) the total assets of $v$ are maximal or (2) the sum of total assets is maximal. 
\end{corollary}

\paragraph{Semi-Positive Swaps}
In contrast to $v$-improving swaps, the reduction in Theorem~\ref{them:opt-setCover} does not readily extend to sequences of semi-positive swaps. In the previous section, we proved a number of positive results for semi-positive swaps (and edge-ranking rules).

Here we show hardness of approximation for maximizing assets via choosing an optimal sequence of semi-positive swaps.
Our construction exploits semi-positivity more directly since it restricts the set of reachable networks from a particular initial network -- solving only the \MaxAssetsv\ problem without this restriction is easy in the instances we consider.

\begin{theorem}\label{thm:opt-IS}
    For a given financial network with monotone allocation rules, an efficient clearing oracle and a bank $v$, it is \classNP-hard to compute a sequence of semi-positive debt swaps such that (1) the total assets of $v$ are maximal and (2) the sum of total assets is maximal.
\end{theorem}
\begin{proof}
\begin{figure}
\centering
\vspace{-2cm}
\hspace{-2cm}
\resizebox{!}{0.6\textwidth}{
	\begin{tikzpicture}[>=stealth', shorten >=1pt, auto,
		node distance=1cm, scale=1, 
		transform shape, align=center, 
		bank/.style={circle, draw, inner sep=1pt}]
		\node[bank] (x1) at (0,0) {$x_1$};
		\node[bank] (x2) [above=of x1] {$x_2$};
		\node (xdots) [above=0.25cm of x2] {$\vdots$};
		\node[bank] (xk) [above=0.75cm of udots] {$x_k$};
		\node[bank] (v1) [right=of x1] {$v_1$};
	 	\node[bank] (v2) [right=of x2] {$v_2$};
		\node (vdots) [above=0.5cm of S2] {$\vdots$};
		\node[bank] (vk) [right=of xk] {$v_k$};
		\node[bank] (e1) at (5,-1) {$e_1$};
		\node[bank] (e2) [above=of e1] {$e_2$};
		\node[bank] (e3) [above=of e2] {$e_3$};
		\node (edots) [above=0.5cm of e3] {$\vdots$};
		\node[bank] (el) [above=0.5cm of edots] {$e_{\ell}$};
		\node[circle, draw, inner sep=2.5pt] (u) [right=7cm of x2] {$u$};
		\node (ancher) [right=of u] {};
		\node (ydots) [above=0.25cm of ancher] {$\vdots$};
		\node[bank] (y2) [below=0.25cm of ancher] {$y_2$};
		\node[bank] (y1) [below=of y2] {$y_1$};
		\node[bank] (yk) [above=of ydots] {$y_k$};
		\node[circle, draw, inner sep=2.5pt] (v) [above=1.5cm of el] {$v$};
		\node[rectangle, draw=black, inner sep=2.5pt] (xy1) [right=0.1cm of y1] {\small{$d_{v_1}+1$}};
		\node[rectangle, draw=black, inner sep=2.5pt] (xy2) [right=0.1cm of y2] {\small{$d_{v_2}+1$}};
		\node[rectangle, draw=black, inner sep=2.5pt] (xyk) [right=0.1cm of yk] {\small{$d_{v_k}+1$}};
		\draw[->] (x1) to node[below, midway] {\footnotesize{$d_{v_1}+1$}} (v1);
		\draw[->] (x2) to node[below, midway] {\footnotesize{$d_{v_2}+1$}} (v2);
		\draw[->] (xk) to node[below, midway] {\footnotesize{$d_{v_k}+1$}} (vk);
		\draw[->] (v1) to node[below, pos=0.2] {\footnotesize{$1$}} (e1);
		\draw[->] (v1) to node[below=0.1cm, pos=0.2] {\footnotesize{$1$}} (e3);
		\draw[->] (v2) to node[below, pos=0.1] {\footnotesize{$1$}} (e1);
		\draw[->] (v2) to node[above, pos=0.2] {\footnotesize{$1$}} (e2);
		\draw[->] (v2) to node[above, pos=0.2] {\footnotesize{$1$}} (el);
		\draw[->] (vk) to node[below=0.2cm, pos=0.1] {\footnotesize{$1$}} (e2);
		\draw[->] (vk) to node[above=0.1cm, pos=0.2] {\footnotesize{$1$}} (e3);
		%\draw[->] (vk) to node[above=0.1cm, pos=0.2] {\footnotesize{$1$}} (el);
		\draw[->] (e1) to node[below=0.2cm, pos=0.2] {\footnotesize{$1$}} (u);
		\draw[->] (e2) to node[below=0.1cm, pos=0.2] {\footnotesize{$1$}} (u);
		\draw[->] (e3) to node[below=0.1cm, pos=0.2] {\footnotesize{$1$}} (u);
		\draw[->] (el) to node[below=0.2cm, pos=0.2] {\footnotesize{$1$}} (u);
		\draw[->] (y1) to node[left=0.1cm, pos=0.3] {\footnotesize{$d_{v_1}+1$}} (u);
		\draw[->] (y2) to node[above=0.1cm, pos=0.1] {\footnotesize{$d_{v_2}+1$}} (u);
		\draw[->] (yk) to node[right, pos=0.3] {\footnotesize{$d_{v_k}+1$}} (u);
		\draw[->] (vk) to node[above, midway] {\footnotesize $1$} (v);
		\draw[->] (v2) to[in=170, out=150, looseness=2] node[left, midway] {\footnotesize $1$} (v);
		\draw[->] (v1) to[in=140, out=160, looseness=2.2] node[left, midway] {\footnotesize $1$} (v);
		\draw[->] (v) to[bend left] node[above=0.1cm, midway] {\footnotesize $k$} (u); 
	\end{tikzpicture}}
\caption{Example construction for reduction with \IS\ where edge labels indicate liabilities.}
\label{fig:IS}
\end{figure}
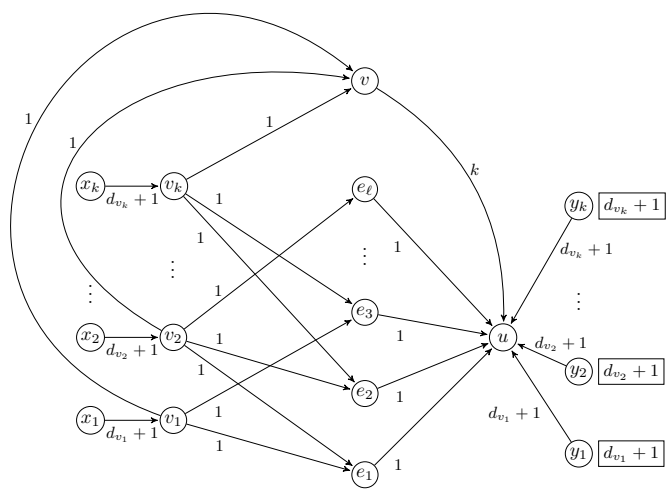
Towards (1), consider an instance of the \IS\ problem with $k$ nodes and $\ell$ edges. \changed{The objective is to derive a maximum set of nodes $\mathcal{I}$ such that no nodes in $\mathcal{I}$ are connected with an edge.}
We construct a financial network with banks $v_1,v_2,\dots,v_k$ representing the nodes and banks $e_1,e_2,\dots,e_\ell$ representing the edges in the given instance. Additionally, add banks $v,u,x_1,x_2,\dots,x_k$ and $y_1,y_2,\dots,y_k$. 
Further, if node $v_i$ is incident to edge $e_j$ in the input instance then include $(v_i,e_j)$ with weight 1 in the financial network. Additionally, every bank $v_i$ has an edge $(v_i,v)$ with weight 1. Bank $u$ has incoming edges $(e_j,u)$ with weight $1$ for all $j \in [\ell]$ and $(v,u)$ with weight $k$. Finally, add edges $(y_i,u)$ and $(x_i,u_i)$ for $i \in [k]$ each with weight $d_{v_i}+1$, respectively, where $d_{v_i}$ is the degree of $v_i$ in the input. Set the external assets of every bank $y_i$ to $d_{v_i}+1$. An example of the construction can be seen in Figure~\ref{fig:IS}.

It can be easily verified that in all networks constructed via semi-positive swaps, every bank $v_i$ with at least two outgoing edges will either have total assets of 0 or receive sufficient assets to clear all debt. As a consequence, the clearing state of every network can be computed efficiently and is independent of the allocation rule. Thus, the result applies for all monotone allocation rules.

We now show that there exists an independent set with cardinality $c$ if and only if there exists a sequence of semi-positive swaps with resulting total assets of $c$ for $v$.
Suppose there exists an independent set $I$ with cardinality $c$. Swap the edges $(x_i,v_i)$ with $(y_i,u)$, for each $i \in I$. Then, every bank $v_i$ pays assets of 1 to $v$ and distributes the remaining assets of $d_{v_i}$ fully among banks $e_j$. Since $I$ is an independent set, the payments by each $v_i$ get fully forwarded to $u$ via the incident banks $e_j$. As a result, all external assets are still forwarded to $u$ and hence all performed swaps are semi-positive. At the same time, $v$ receives total assets of $c$.

For the other direction assume there exists a sequence of semi-positive debt swaps such that the resulting total assets of $v$ are given by $c$. First, let us observe that no swap of edges with weight 1 can be semi-positive. The only debt swap where all creditors and debtors are pairwise distinct is given by two edges $(v_i,e_j)$ and $(v_{i'},e_{j'})$. However, since such a swap cannot create a new cycle and there exists no path from $e_j$ to $e_{j'}$ and vice versa, no bank can strictly profit while the other bank is not harmed. Hence, it suffices to restrict our attention to swaps between $u$ and a bank $v_i$.

Bank $v$ can only receive payments of 1 from each bank $v_i$. Thus, at least $c$ of the banks $v_i$ must have positive total assets. This implies that $u$ and at least $c$ banks $v_i$ must have executed semi-positive swaps. This only happens if all such banks $v_i$ forward all assets to $u$. As a direct implication, no two distinct banks $v_i$ and $v_{i'}$ pay assets to the same bank $e_j$. In conclusion, the banks represent an independent set of cardinality $c$ in the input instance.

For statement (2) we slightly adapt the construction. Replace the edge $(v,u)$ by a path of $B+1$ edges and $B$ additional banks $w_1,w_2,\dots,w_B$. Assign each edge $(w_i,w_{i+1})$ weight $B+i$, edge $(v,w_1)$ with weight $B$ and $(w_B,u)$ with weight $2B$. Indeed, it suffices to set $B$ to a sufficiently large value of, say, $B = 6k^2$. When executing one additional swap with a node $v_i$ and $u$, the sum of assets of banks $w_i$ increases by at least $6k^2$ whereas the sum of assets in the remaining graph is upper bounded by $5k^2+k$. For this reason, it is always optimal to perform debt swaps to maximize the assets of $v$ and, thus, the cardinality of the corresponding independent set. Maximizing the sum of total assets reduces to maximization of assets of $v$. 
\end{proof}

The \IS\ problem cannot be approximated within $\Omega(k^{1-\epsilon})$, unless \classP=\classNP\, where $k$ is the number of nodes~\cite{hastad1999clique}. In our reduction, $v$ has total assets of $c$ if and only if there exists an independent set of size $c$ in the given \IS\ instance. If the input instance of \IS\ consists of $k$ nodes, then the number $n$ of banks in the construction is bounded by $O(k^2)$. Therefore, it is straightforward to argue that, for any constant $\varepsilon > 0$, we obtain \classNP-hardness of computing a $O(n^{1/2-\varepsilon})$-approximation to the maximal total assets of a given bank $v$ that can be achieved using a sequence of semi-positive debt swaps. The argument applies also when maximizing the sum of all assets, since the total number of banks in the resulting network is still bounded by $O(k^2)$.
\begin{corollary}\label{cor:opt-IS}
    For a given financial network with monotone allocation rules, an efficient clearing oracle and a bank $v$, it is \classNP-hard to compute an $n^{1/2 - \varepsilon}$-approximation to (1) the maximal total assets of $v$ and (2) the maximal sum of total assets that can be obtained by a sequence of semi-positive debt swaps, for every constant $\varepsilon > 0$.
\end{corollary}

The next theorem shows that the problem remains \classNP-hard even under the restriction that $v$ is involved in every swap and every semi-positive swap is also $v$-improving.

\begin{theorem}\label{thm:opt-2partition}
     For a given financial network with monotone allocation rules, an efficient clearing oracle, and a bank $v$, it is \classNP-hard to compute a sequence of $v$-improving semi-positive debt swaps with creditor $v$ such that (1) the total assets of $v$ are maximal and (2) the sum of total assets is maximal.
\end{theorem}
\begin{proof}
For statement (1), consider an instance of \Tpartition\ with items $i \in [k]$ and integer values $a_i \in \NN$. Moreover, denote the sum of values by $A$, i.e., $A = \sum_{i \in [k]} a_i$. \changed{The goal is to divide the items into two sets $\mathcal{A}_1$ and $\mathcal{A}_2$ such that $\mathcal{A}_1 \cap \mathcal{A}_2 = \emptyset$, $\mathcal{A}_1 \cup \mathcal{A}_2 = [k]$ and $\sum_{i \in \mathcal{A}_1} a_i = \sum_{i \in \mathcal{A}_2} a_i$.}

For our reduction, we define a financial network with banks $v,u,a_1,a_2,\dots,a_k$ and $k$ auxiliary banks $s_1,s_2,\dots,s_k$. For every $i \in [k]$ there are edges $(a_i,u)$ and $(s_i,v)$ with weights $M>\max_{j \in [k]} a_j$. There is a single edge $(v,u)$ with weight $A/2$. Moreover, bank $a_i$ has external assets $a_i$, for every $i \in [k]$. Clearly, the construction can be computed efficiently in the size of the instance of \Tpartition.

Every bank has at most one outgoing edge and thus pays all assets towards that edge until the debt is settled. The clearing state is the same for every monotone allocation rule and can be computed efficiently. 

We show that there the instance of \Tpartition\ has a solution if and only if there is a sequence of $v$-improving semi-positive debt swaps with creditor $v$ such that the total assets of $v$ are $A/2$.

First, assume there exists a partition in two sets $\mathcal{A}_1,\mathcal{A}_2 \subset [k]$ such that $\mathcal{A}_1 \cup \mathcal{A}_2 = [k]$, $\mathcal{A}_1 \cap \mathcal{A}_2 = \emptyset$ and $\sum_{i \in \mathcal{A}_1} a_i = A/2 = \sum_{i \in \mathcal{A}_2} a_i$. Hence, when performing debt swaps of edges $(a_i,u), (s_i,v)$, for $i \in \mathcal{A}_1$, then $v$ has total assets of $A/2$. By construction, all assets of $v$ are payed to $u$. Thus, all swaps are semi-positive.

For the other direction, assume there exists a sequence of $v$-improving semi-positive debt swaps with creditor $v$ such that the total assets of $v$ are given by $A/2$. Then clearly, there exists a set of banks $\mathcal{A} \in [k]$ where total assets add up to $A/2$. There exists a partition $\mathcal{A}, \bar{\mathcal{A}}=\{i \notin \mathcal{A} \mid i \in [k]\}$ for the given input.

Statement (2) follows when including $B$ additional banks $w_b$ on the path from $v$ to $u$ with edge weight $A/2$ for edge $(v,w_1)$ and $B+b$ for every edge $(w_b,\cdot)$ with $b \in [B]$. This construction does not allow any additional debt swaps. Since the total assets of every bank $w_b$ equal those of $v$, maximizing the sum of total assets reduces to maximizing assets of $v$ for sufficiently large $B$.
\end{proof}

For financial networks with edge-ranking rules we show a slightly stronger result.

\begin{theorem}\label{thm:opt-3partition}
    For a given financial network with edge-ranking rules and a bank $v$, it is strongly \classNP-hard to compute a sequence of $v$-improving semi-positive debt swaps with creditor $v$ such that (1) the total assets of $v$ are maximal and (2) the sum of total assets is maximal.  
\end{theorem}
\begin{proof}
\begin{figure}[t]
\centering
\resizebox{!}{0.32\textwidth}{
    \begin{tikzpicture}[>=stealth', shorten >=1pt, auto,
    node distance=1.5cm, scale=1, 
    transform shape, align=center, 
    bank/.style={circle, draw, inner sep=1}]
    \node[bank] (u1) at (0,0) {$u_1$};
	\node[bank] (u2) [right = of u1] {$u_2$};
	\node[bank] (u3) [right = of u2] {$u_3$};
	\node[bank] (u4) [right = of u3] {$u_4$};
	\node (mdots) [right = of u4] {$\cdots$};
	\node[bank] (ul) [right = of mdots] {$u_l$};
	\node[circle, draw, inner sep=2.5pt] (v) [above =1cm of u4] {$v$};
	\node[circle, draw, inner sep=2.5pt] (r) [below =1cm of u4] {$r$};
	\node[bank] (s1) [right =0.7cm of u2] {$s_2$};
	\node[bank] (s2) [right =0.7cm of u4] {$s_4$};
	\node[rectangle, draw=black, inner sep=2.5pt] (xs1) [below=0.05cm of s1] {\footnotesize{$1$}};
	\node[rectangle, draw=black, inner sep=2.5pt] (xs2) [below=0.05cm of s2] {\footnotesize{$1$}};
	\node[bank] (a3) [below=0.5cm of r] {$a_3$};
	\node[bank] (a2) [left=1cm of a3] {$a_2$};
	\node[bank] (a1) [left=1cm of a2] {$a_1$};
	\node (ldots) [right=0.5cm of a3] {$\cdots$};
	\node[bank] (ak) [right=0.5cm of ldots] {$a_{k}$};
	\node[rectangle, draw=black, inner sep=2.5pt] (xa1) [below=0.05cm of a1] {\footnotesize{$a_1$}};
	\node[rectangle, draw=black, inner sep=2.5pt] (xa2) [below=0.05cm of a2] {\footnotesize{$a_2$}};
	\node[rectangle, draw=black, inner sep=2.5pt] (xa3) [below=0.05cm of a3] {\footnotesize{$a_3$}};
	\node[rectangle, draw=black, inner sep=2.5pt] (xak) [below=0.05cm of ak] {\footnotesize{$a_k$}};
	\node (anchor) [above=0.5cm of v] {};
	\node[bank] (axk) [left=0.5cm of anchor] {};
	\node (adots) [left=0.5cm of axk] {$\cdots$};
	\node[bank] (ax2) [left=0.5cm of adots] {};
	\node[bank] (ax1) [left=1cm of ax2] {};
	\node[bank] (sx1) [right=0.5cm of anchor] {};
	\node[bank] (sx2) [right=1cm of sx1] {};
	\node (sdots) [right=0.5cm of sx2] {$\cdots$};
	\node[bank] (sxl) [right=0.5cm of sdots] {};
	\draw[->] (u1) to node[left=0.4cm, midway] {\footnotesize{$M$}} (r);
	\draw[->] (u3) to node[left=0.1cm, midway] {\footnotesize{$M$}} (r);
	\draw[->] (ul) to node[left=0.2cm, midway] {\footnotesize{$M$}} (r);
	\draw[->] (v) to node[left=0.3cm, pos=0.7] {\footnotesize{\textcolor{blue}{$1$},$T$}} (u1);
	\draw[->] (v) to node[left=0.2cm, pos=0.7] {\footnotesize{\textcolor{blue}{$2$},$1$}} (u2);
	\draw[->] (v) to node[left=0.1cm, pos=0.7] {\footnotesize{\textcolor{blue}{$3$},$T$}} (u3);
	\draw[->] (v) to node[left, pos=0.7] {\footnotesize{\textcolor{blue}{$4$},$1$}} (u4);
	\draw[->] (v) to node[left=0.3cm, pos=0.7] {\footnotesize{\textcolor{blue}{$l$},$T$}} (ul);
	\draw[->] (s1) to node[below, pos=0.4] {\footnotesize{$d$}} (u2);
	\draw[->] (s2) to node[below, pos=0.4] {\footnotesize{$d$}} (u4);
	\draw[->] (a1) to node[left=0.2cm, midway] {\footnotesize{$c$}} (r);
	\draw[->] (a2) to node[left=0.1cm, midway] {\footnotesize{$c$}} (r);
	\draw[->] (a3) to node[left, midway] {\footnotesize{$c$}} (r);
	\draw[->] (ak) to node[left=0.2cm, midway] {\footnotesize{$c$}} (r);
	\draw[->] (ax1) to node[left=0.3cm, midway] {\footnotesize{$c$}} (v);
	\draw[->] (ax2) to node[right=0.2cm, midway] {\footnotesize{$c$}} (v);
	\draw[->] (axk) to node[right, midway] {\footnotesize{$c$}} (v);
	\draw[->] (sx1) to node[right, midway] {\footnotesize{$d$}} (v);
	\draw[->] (sx2) to node[right=0.1cm, midway] {\footnotesize{$d$}} (v);
	\draw[->] (sxl) to node[right=0.2, midway] {\footnotesize{$d$}} (v);
\end{tikzpicture}}
\caption{Schematic construction for the reduction from \partition\ where blue edge labels indicate edge-ranking rules and black labels indicate edge weights.}
\label{fig:3-partition}
\end{figure}
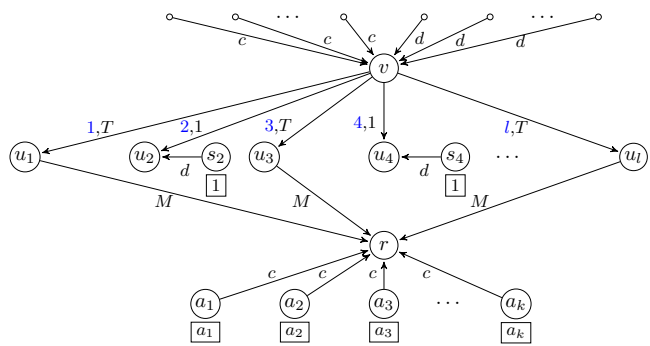
We show statement (1) by reduction from \partition. An instance of \partition\ consists of a set of $k$ elements such that $k$ is an integer multiple by $3$. Each item $i \in [k]$ has a value $a_i \in \NN$. We denote the sum of values by $A = \sum_{i \in [k]} a_i$ and let $T =  \frac 3k A$. For each value it holds that $\frac T4 < a_i < \frac T2$. The goal is to decide whether or not the given set of $k$ items can be partitioned into $k$ sets $S_1,\ldots,S_k$ with $\sum_{i \in S_j}a_i = T$, for each $j \in [k]$. Due to the value restriction, each set must have exactly three items.

For an instance of \partition, we construct a financial network as follows. First, add a root $r$ and a bank $a_i$ with external assets $a_i$, for every item $i$. Additionally, every bank $a_i$ has a single outgoing edge $(a_i,r)$ with edge weight $c \geq \frac T2$. Construct bank $v$ and $u_1, u_2, \dots, u_l$, for $l = 2\frac k3-1$, together with edges $(v,u_h)$. For all odd numbers $h$, set the weight of $(v,u_h)$ to $T$ and add the edge $(u_h,r)$ with weight $M > T$. Instead, set the weight of $(v,u_h)$ to $1$, for all even numbers $h$, and include a source $s_h$ with external assets of $1$ and edge $(s_h,u_h)$ and weight $d>1$. Finally, add $k$ auxiliary banks, each with a single outgoing edge to $v$ and weight $c$, and $\frac l2$ auxiliary banks with outgoing edge to $v$ and weight $d>1$. The values $T,M,d$ and $c$ are chosen to be distinct. $v$ is the only bank with more than one outgoing edge, and it distributes all assets according to an edge-ranking rule with $\pi_v(h)=u_h$ for all $h$. That means, first $v$ fully pays off debt to $u_1$, then to $u_2$,..., until $v$ runs out of funds. A schematic construction can be seen in Figure~\ref{fig:3-partition}. 

Note that every bank has at most one outgoing edge. As such, the clearing state is independent of the allocation rule and can be determined easily.

We show that the \partition\ instance has a solution if and only if there is a sequence of $v$-improving semi-positive debt swaps resulting in total assets of at least $A+\frac l2-1$ for $v$. 

First, assume there exists a solution to the \partition\ instance. Then, swap all edges incident to banks $a_i \in S_1$ with incident edges of $v$. This results in total assets of $T$ for $v$ and thus she can exactly pay off $(v,u_1)$. Clearly, the described swaps are all semi-positive and $v$ is the benefiting creditor. Now, execute a swap with creditors $v$ and $u_2$. Since, this results in additional assets of $1$ in $v$ and the second edge has payments zero and is active, the swap is semi-positive. When repeatedly applying this process for the remaining sets $S_j$, bank $v$ has swapped all incoming edges in the initial network. Consequently, the total assets of $v$ are given by $A+ \frac l2 -1$.

For the other direction, assume there exists a sequence of $v$-improving semi-positive debt swaps resulting in total assets of at least $A+ \frac l2 -1$ for $v$. Because the graph is cycle-free, all assets must originate from external assets. In particular, observe that the sum of external assets in the network is given by $A+ \frac l2 -1$ and hence all banks with external assets must pay all assets to $v$. Clearly, this is only the case if all banks with external assets are swapped to be adjacent to $v$. 
In the initial network, due to the allocation rule $v$ can only perform semi-positive swaps with $r$. Moreover, $v$ and $r$ can switch edges as long as the residual capacity of $(v,u_1)$ is sufficiently large. Only when the total assets of $v$ are exactly $T$, can $v$ fully pay off the first edge towards $u_1$ and execute a semi-positive swap with $u_2$. Otherwise, a fraction of the additional assets would either be payed to $u_1$ or $u_3$. These arguments can be repeated for all remaining outgoing edges of $v$. As a result, $v$ and $r$ can only swap all their initial incoming edges if the external assets $a_i$ can be grouped into sets with sum of value exactly $T$. Since $\frac T4 < a_i < \frac T2$, every set has cardinality three. 

Statement (2) follows by the same construction as for (1). We claim that there exists a solution to the given instance of \partition\ if and only if there exists a sequence of $v$-improving semi-positive debt swaps resulting in at least $4A+3(\frac l2 -1)$ sum of total assets.

If there exists a solution, then proceed as for (1). Clearly, the total assets of all sources add up to $A+\frac l2 -1$ and the total assets of $v$ are given by same amount. Additionally, $v$ pays all assets to banks $u_h$ whereas all payments of $A$ are forwarded to $r$. In total, assets of $4A+3(\frac l2 -1)$ are achieved.

Suppose there exists a sequence of semi-positive swaps resulting in a sum of total assets of at least $4A+3(\frac l2 -1)$. Clearly, this value can only be achieved when $v$ settles all debt.
By the same argument as before, this directly implies existence of a solution to \partition.     
\end{proof}

Observe that our constructions in the proofs of Theorems~\ref{them:opt-setCover},~\ref{thm:opt-IS},~\ref{thm:opt-2partition} and~\ref{thm:opt-3partition} lead to classes of cycle-free financial networks. Therefore, all \classNP-hardness results in this section apply even under the restriction that the underlying graph structure of the financial network is a DAG. 

%%%%%%%%%%%%%%%%%%%%%%%%%%%%%

\section{Reachability}\label{sec:reach}

A debt swap changes the underlying graph structure of the corresponding financial network. It is a natural reachability question whether, for two given networks, one network can be transformed into the other via a sequence of debt swaps.

More formally, in the \Reach\ problem we are given an initial network $\F$ and a target network $\G$. Both networks have the same set $V$ of nodes and are consistent with the stubs from an instance of \MaxAssetsv. 
We say every node $v \in V$ has the same \emph{incidence profile} in both networks (i.e., for each value $l_i$ the same number of outgoing edges with weight $l_i$ in both $\F$ and $\G$ and the same number of incoming edges with $l_i$ in both $\F$ and $\G$). Moreover, $\F$ and $\G$ have the same monotone allocation rule, e.g., $v$ uses the same edge ranking to allocate assets to outgoing stubs, for each $v \in V$. Overall, we say the networks $\F$ and $\G$ are \emph{consistent with each other}. 
If $\F$ and $\G$ are not consistent with each other, then the answer to every reachability problem is trivially negative. Hence, we assume that $\F$ and $\G$ are consistent with each other. To avoid trivialities, $\F$ and $\G$ have different edge sets.

For a finite sequence of debt swaps $\bsigma=(\sigma_1,\sigma_2,\dots)$ to $\F$, we first execute swap $\sigma_1$ in $\F$ and denote the resulting network by $\F^{\sigma_1}$. We again assume to have access to an efficient clearing oracle $\mathcal{O}$ to compute the clearing state. Then perform $\sigma_2$ in $\F^{\sigma_1}$, then $\sigma_3$ in $\F^{\sigma_2}$, etc. The process ends after all swaps in the sequence were applied. The resulting network is denoted by $\F^{\bsigma}$. Every financial network created during this sequence is called an \emph{intermediate} network. We say that $\G$ is \emph{reachable} from $\F$ if there exists a sequence of debt swaps $\bsigma$ such that $\F^{\bsigma}=\G$. The sequence $\bsigma$ is called a \emph{reaching} sequence.
In the \Reach\ problem, we are interested in reaching sequences of \emph{arbitrary} swaps. If we only consider sequences of $v$-improving arbitrary debt swaps, then the problem is denoted by \Reachv. Intuitively, instead of strict profits for a given bank $v$ it might be sufficient that $v$ stays solvent. Formally, in the \Reachvl\ problem we are given consistent networks $\F$ and $\G$ such that $v$ has the \emph{same} amount $\ell$ of total assets in both networks. We shall decide if there exists a reaching sequence of debt swaps such that the total assets of $v$ are at least $\ell$ in every intermediate network.

These problems can be interpreted as instances of the $s$-$t$-Connectivity (\textsc{Stcon}) problem in a transition graph. The nodes of this transition graph are the networks consistent with $\F$ and $\G$. There is a directed edge between networks $\F'$ and $\F''$ if and only if there is a debt swap that turns $\F'$ into $\F''$. \textsc{Stcon} is solvable in poly-logarithmic space~\cite{Savitch70}. In our case, there are less than $(m!)^{n/(2m)}<(\frac m 2)^m$ consistent networks~\cite{alon2008maximum}, where the inequality holds for large $m$. As such, all our reachability problems lie in \classPSPACE.

Our first result shows that for arbitrary swaps, the problem is much easier. The consistency property is necessary \emph{and sufficient} for reachability. A reaching sequence of polynomially many swaps can be computed in polynomial time if it exists.

\begin{theorem}\label{thm:greedy}
    The \Reach\ problem with monotone allocation rules can be solved in time polynomial in the input size.
\end{theorem}
\begin{proof}
    Let $\mathcal{D}$ be the set of desired edges that are in $\G$ and not in $\F$, i.e., $\mathcal{D}=E^{\G}\setminus E^{\F}$. There are exactly $|\mathcal{D}| > 0$ edges in $E^{\F}\setminus E^{\G}$. First, observe that $|\mathcal{D}|>1$. To see this, assume for a contradiction that $\mathcal{D} =\{e\}$ for $\deb(e)=u$, $\cre(e)=v$ and $u,v \in V$ with weight $l_e>0$. Since $\F$ and $\G$ are consistent with each other, in $\F$ bank $v$ must have an incoming edge $e'$ with $\deb(e')=x$, and $u$ must have an outgoing edge $e''$ with $\cre(e'')=y$ each with weight $l_e$, and both these edges are not in $E^{\G}$. Since $|\mathcal{D}|=1$, banks $x$ and $y$ have more incident edges in $\F$ than $\G$, a contradiction. We assume $|\mathcal{D}|\geq 2$ for the remaining proof.

    We construct a sequence of debt swaps that in each step replaces an edge in $E^{\F}\setminus E^{\G}$ by one in $\mathcal{D}$.
    Choose an arbitrary edge $\hat{e} \in \mathcal{D}$ with $\deb(\hat{e})=w, \cre(\hat{e})=v$. By consistency there exists an edge $e \in E^{\F} \setminus E^{\G}$, where $\deb(e)=u,\cre(e)=v$, with weight $l_e=l_{\hat{e}}$. Further, $w$ has at least one outgoing edge $e'$ with $\cre(e')=x$ in $\F$ that is not in $\G$ with weight $l_{e'}=l_e$. 
    If the banks $u,v,w$ and $x$ are pair-wise distinct, choose another arbitrary edge in $\mathcal{D}$. Clearly, there always exists such an edge if $\G$ is reachable from $\F$. 
    Otherwise, perform a swap $\sigma$ of edges $e$ and $\hat{e}$. The resulting network $\F^{\sigma}$ contains the desired edge $\hat{e} \in \mathcal{D}$ connecting $w$ and $v$ and one edge in $E^{\F}\setminus E^{\G}$ was removed. We repeat this step until two desired edges in $\mathcal{D}$ have not been included. 
    
    For the final step with $|\mathcal{D}|=2$, observe that there are exactly two edges in $E^{\F}\setminus E^{\G}$. Choose an arbitrary edge $\hat{e} \in \mathcal{D}$ with $\deb(\hat{e})=w, \cre(\hat{e})=v$. By the same argument as before, there exist edges $e, e' \in E^{\F}\setminus E^{\G}$ where $\deb(e)=u,\cre(e)=v$ and $\deb(e')=w,\cre(e')=x$ each with the same weight as $\hat{e}$. By consistency, bank $u$ must have an outgoing edge with weight $l_{\hat{e}}$ that is not in $\F$. $x$ must have an incoming edge with weight $l_{\hat{e}}$ that is not in $\F$. Since $|\mathcal{D}|=2$, this implies the existence of edge from $u$ to $x$ in $\G$. Then, $\G$ is reachable from $\F$ with the debt swap of edges $e$ and $e'$ creating the edges in $\mathcal{D}$.

    This proves that the network $\G$ can indeed be reached via a sequence of debt swaps. In every step of this sequence, a debt swap of edges not in $\G$ is performed to create a new desired edge in $\mathcal{D}$. This means that a desired edge created in a previous step will never be swapped again. Thus, the number of swaps in the resulting sequence is bounded by $O(m)$ where $m=|E^{\F}|=|E^{\G}|$ and $m \leq|\mathcal{D}|$.  
\end{proof}

The greedy algorithm described in the proof of Theorem~\ref{thm:greedy} arbitrarily chooses an edge from the target network and performs a debt swap to include the edge into the current network. Note that this can lead to devastating losses of assets of a single bank $v$ in the intermediate networks, which is not aligned with incentives and potentially problematic for the financial stability of the system.
\changed{Instead, for $v$-improving swaps, in which a given bank $v$ must strictly profit in every step, the tight \classPLS-reduction in Theorem~\ref{thm:PLS} directly implies the following result~\cite[Theorem 16]{Yannakakis03}.}

\begin{corollary}
    \label{cor:PSPACE}
    The \Reachv\ problem with monotone allocation rules is \classPSPACE-complete.
\end{corollary}

We strengthen this result in the following way. Instead of guaranteeing a strict profit for $v$ in every step, we require only that the total assets of $v$ are at least $\ell$ in the initial, final, and every intermediate network. Deciding whether there exists a reaching sequence of arbitrary debt swaps so that $v$ \emph{never has strictly less assets} in any step is \classPSPACE-hard (and, hence, \classPSPACE-complete).
\begin{restatable}{theorem}{thmPSPACE} \label{thm:PSPACE}
    The \Reachvl\ problem with monotone allocation rules is \classPSPACE-comple\-te.
\end{restatable}
\begin{proof}
    An instance of \satcon\ consists of a CNF formula $\varphi$ with $k$ variables $x_i$ and $\ell$ clauses where every clause $\kappa_j$ consists of at most three variables. Every satisfying assignment of $\varphi$ is represented by a node in the graph $G(\varphi)$ where two assignments are adjacent if the assignments differ for exactly one variable. For two given assignments, the objective is to decide whether there exists a path connecting the corresponding nodes in $G(\varphi)$. Equivalently, this requires to decide whether there is a sequence of flips between these satisfying assignments changing the truth assignment of exactly one variable in each step. The authors in~\cite{GopalanSatConnectivity} show \classPSPACE-completeness for \satcon.
    
    For a reduction, use a slight adaptation of the construction illustrated in \changed{Figure~\ref{fig:Max2Sat}} described in proof of Theorem~\ref{thm:PLS} where all values $w_j$ are set to one. Also, choose edges $(x_i,x_i^T)$ and $(s_i,x_i^F)$ if variable $x_i$ is TRUE in the initial assignment, otherwise choose $(x_i,x_i^F)$ and $(s_i,x_i^T)$. Observe that every debt swap switches the incoming edges of a pair of banks $x_i^T$ and $x_i^F$ where the bank with incoming payments from $x_i$ is interpreted as TRUE. As a result, every sequence of flips in the assignment translates directly to a sequence of debt swaps and vice versa.
    
    The total assets of bank $v$ equal the number of banks $\kappa_j$ with incoming assets. Hence, if $v$ has assets of at least $l$ then every $\kappa_j$ has incoming payments and the associated assignment is satisfying. As an implication, there exists a sequence of flips in the assignment, where every intermediate assignment is satisfying, if and only if there exists a sequence of debt swaps turning the initial into the target network, where $v$ has assets of $\ell$ in every intermediate network. This concludes the proof. 
\end{proof}

Finally, we briefly consider sequences of semi-positive debt swaps. Indeed, since they are Pareto-improving, the motivation here is similar as in the previous theorem. Deciding whether a reaching sequence exists where no bank is harmed is already challenging, even for edge-ranking rules.

\begin{restatable}{proposition}{propReachNP} \label{prop:reach-NP}
    Given initial network $\F$ and a target network $\G$ with edge-ranking rules that are consistent with each other, it is strongly \classNP-hard to decide whether $\G$ is reachable from $\F$ via any sequence of semi-positive swaps.
\end{restatable}

\begin{proof}
    The statement follows by reduction from \partition. As initial network we choose the construction in proof of Theorem~\ref{thm:opt-3partition} (see Figure~\ref{fig:3-partition}). The target is the network where all incoming edges with weight $c$ of $v$ and $r$ are switched.
    By the same arguments as in proof of Theorem~\ref{thm:opt-3partition} there exists a desired sequence of semi-positive debt swaps if and only if there exists a solution to the instance of \partition.
\end{proof}

%%%%%%%%%%%%%%%%%%%%%%%%%%%%%

\section{Conclusion}
In this paper we have studied the computational complexity of problems surrounding the debt swap operation, a natural and interesting edge swap operation in financial networks. Our results characterize the convergence time of natural debt swap dynamics for edge-ranking payment rules. Moreover, our work sheds light on the complexity of optimization and reachability properties of debt swap dynamics.

Our work gives rise to a plethora of interesting future directions. Most prominently, the general \SequenceSwapSP\ problem remains a challenging open problem. Also, it would be interesting to identify additional meaningful cases in which reachability or optimization problems can be solved (approximately) optimally in polynomial time. More generally, debt swaps can lead to interesting subtle structural effects. For example, a sequence of debt swaps from a network $\F$ with edge-ranking rules can result in a financial network $\G$ with the same underlying graph but a different edge-ranking rule. Understanding the structural, algorithmic, and economic consequences of these effects represent fascinating open problems.

%%%%%%%%%%%%%%%%%%%%%%%%%%%%%

\subsection*{Acknowledgments}
Martin Hoefer gratefully acknowledges support from DFG Research Unit ADYN under grant DFG 411362735 and ATHENE National Research Center for Applied Cybersecurity.
Lisa Wilhelmi gratefully acknowledges support from DFG Research Unit ADYN under grant DFG 411362735.

%%%%%%%%%%%%%%%%%%%%%%%%%%%%%%%%%%%%%%%%%%%%%%%%%%%%%%%%%%%%%%%%%%%%%%%%%%%%%%%%%%%%%%%%%%%%%%%%%%%%%%%%%%%%%%%%%%%%%

% Bibliography
\bibliographystyle{plainurl}
\bibliography{ref}

\clearpage

% Appendix

\appendix

%%%%%%%%%%%%%%%%%%%%%%%%%%%%%%%%%%%%%%%%%%%%%%%%%%%%%%%%%%%%%%%%%%%%%%%%%%%%%%%%%%%%%%%%%%%%%%%%%%%%%%%%%%%%%%%%%%%%%

\section{Missing Proofs in Section~\ref{sec:characterization}}
\label{app:proofs}

\setcounter{lemma}{0}
\renewcommand{\thelemma}{\Alph{section}.\arabic{lemma}}
\renewcommand{\thedefinition}{\Alph{section}.\arabic{lemma}}
\renewcommand{\thecorollary}{\Alph{section}.\arabic{lemma}}

\begin{definition}
Consider a financial network $\F$ with external assets $\veca^x$ and a non-negative vector $\hat{\veca}^x \leq \veca^x$, where $\hat{a}^x_v < a^x_v$ for at least one $v\in V$. The network $\F$ can be separated into two networks: $\pre{\F} = (G, \deb, \cre, \vecl, \hat{\veca}^x, \vecf)$ is called a \emph{pre-network} of $\F$; $\diff{\F} = (G, \deb, \cre, \diff{\vecl}, \veca^x-\hat{\veca}^x, \diff{\vecf})$ is the \emph{difference network} of $\F$ and $\pre{\F}$. The liabilities in $\diff{\F}$ are given by the remaining liabilities after clearing in $\pre{\F}$, (i.e., $\diff{l}_e = l_e - \pre{p}_e$), and the allocation rule is defined as $\diff{f}_e(b_{\deb(e)}) = f_e(\pre{a}_{\deb(e)} + b_{\deb(e)}) - f_e(\pre{a}_{\deb(e)})$.
\end{definition}

The clearing state in $\F$ equals the sum of clearing states in pre- and difference networks.

\begin{restatable}[Separability]{lemma}{lemSeparability}\label{lem:separability}
    Suppose $\pre{\F}$ is a pre-network of $\F$ and $\diff{\F}$ the corresponding difference network. The clearing state of $\F$ is given by the sum of clearing states in $\pre{\F}$ and $\diff{\F}$, i.e., $\vecp = \pre{\vecp} + \diff{\vecp}$.
\end{restatable}

%\lemSeparability*

\begin{proof}
As mentioned in Section~\ref{sec:model}, the feasible asset vectors for a financial network $\F$ are solutions of a fixed point problem. We reformulate this property for the feasible payments -- they are fixed points of a function $g$ defined for every $e \in E$ by $g_e(\vecc) = f_e(a^x_u + \sum_{e'\in E^-(u)}c_{e'})$, where $u = \deb(e)$. $g$ maps the convex, compact set $\{ \vecc \mid 0 \le c_e \le l_e\}$ to itself. The fixed points of $\pre{g}_e(\vecc)= f_e(\hat{a}^x_u + \sum_{e' \in E^-(u)} c_{e'})$ are the feasible payments in $\pre{\F}$, the fixed points of $\diff{g}_e(\vecc)= \diff{f}_e(a^x_u - \hat{a}^x_u + \sum_{e' \in E^-(u)} c_{e'})$ the feasible payments in $\diff{\F}$.

For the maximal fixed point $\diff{\vecp}$ of $\diff{g}$ it holds that
\begin{align*}
    \diff{g}_e(\vecp) &= \diff{f}_e\left(a^x_u - \hat{a}^x_u + \sum_{e'\in E^-(u)}p_{e'}\right)  \\
               &= f_e\left(\pre{a}_u + a^x_u - \hat{a}^x_u + \sum_{e'\in E^-(u)}p_{e'}\right) - f_e\left(\pre{a}_u\right) \\
               &= f_e\left(\hat{a}^x_u + \sum_{e' \in E^-(u)}\pre{p}_{e'} + a^x_u - \hat{a}^x_u + \sum_{e'\in E^-(u)}p_{e'}\right) - f_e\left(\hat{a}^x_u + \sum_{e' \in E^-(u)}\pre{p}_{e'}\right) \\
               &= f_e\left(a^x_u + \sum_{e'\in E^-(u)}(p_{e'} + \pre{p}_{e'})\right) - \pre{p}_{e} \enspace.
\end{align*}
The second equality follows from the definition of $\diff{f}_e$. The last equality follows since $\pre{\vecp}$ is the maximal fixed point of $\pre{g}$. 

The clearing state of $\F$ satisfies $\vecp \geq \pre{\vecp}$ due to monotonicity of all $f_e$. Indeed, for \emph{every} fixed point $\vecc$ of $g$ it holds that $\vecc \geq \pre{\vecc}$, since $a^x_v > \hat{a}^x_v$ for at least one $v$. Thus, every fixed point $\vecc$ can be separated into $\pre{\vecp}$ and $\bar{\vecc}$ such that $\vecc = \pre{\vecp}+\bar{\vecc}$. Now consider $\bar{\vecc}$. Since $\vecc = \pre{\vecp}+\bar{\vecc}$ is a fixed point, we have $\pre{p}_e+\bar{c}_e = f_e(a^x_u + \sum_{e' \in E^-(u)} (\pre{p}_{e'} + \bar{c}_{e'}))$ for all $e \in E$ and $u = \deb(e)$. Equivalently,
$$
   \pre{p}_e = f_e\left(a^x_u + \sum_{e \in E^-(u)} (\pre{p}_e + \bar{c}_e)\right) - \bar{c}_{e}
$$
and, thus, $\bar{\vecc}$ is a fixed point of $\diff{g}$. Hence, every fixed point of $g$ can be separated into $\pre{\vecp}$ and a fixed point for $\diff{g}$. Since the clearing state $\vecp$ is the \emph{maximal} fixed point of $g$, it must use the \emph{maximal} fixed point of $\diff{g}$, i.e., $\vecp = \pre{\vecp} + \diff{\vecp}$.
\end{proof}

Note that if edge $e\in E$ gets saturated in $\pre{\F}$ (i.e., $l_e=\pre{p}_e$), then $\diff{l}_e=0$ in $\diff{\F}$.

\begin{corollary}\label{cor:pre-network-leq-clearing}
If $\pre{\F}$ is a pre-network of $\F$, then for the clearing payments $p_e \geq \pre{p}_e$, for all $e\in E$, and the total assets $a_v \geq \pre{a}_v$, for every bank $v\in V$.
\end{corollary}

The next definition establishes a relationship between the external assets of a node $v$ and the incoming assets of a sink.

\begin{definition}[Linearity]\label{def:linearity}
If an increase of the external assets $a^x_v$ of a node $v \in V$ by $\Delta > 0$ results in total assets $a'_t = a_t + \Delta$ for a sink $t \in V$, then $v$ is called \emph{$t$-linear} on $[a^x_v, a^x_v + \Delta]$. If the increase leads to $\sum_{t \in U} a'_t = \sum_{t \in U} a_t + \Delta$ for a set $U$ of sink nodes, then $v$ is called \emph{$U$-linear} on $[a^x_v, a^x_v + \Delta]$.
\end{definition}

If $v$ is $t$-linear, then the benefit of $t$ due to the increase of external assets of $v$ is maximal.

\lemMonotonicity*

\begin{proof}
The network $\F$ can be seperated into the pre-network $\pre{\F} = \hat{\F}$ and difference network $\diff{\F}$. The external assets in $\diff{\F}$ are given by $a^x_v-\hat{a}^x_v$ for $v$ and 0 for all other banks $u \neq v, u \in V$. By separability, the clearing state of $\F$ is the sum of clearing states of $\pre{\F}$ and $\diff{\F}$, i.e., $\vecp = \pre{\vecp} + \diff{\vecp}$. Since the total assets of every bank are non-negative, we have $a_u = \pre{a}_u + \diff{a}_u \geq \pre{a}_u = \hat{a}_u$.
\end{proof}

\lemNonExpan*

\begin{proof}
By slight abuse of notation, we use $\F$ to denote the resulting financial network \emph{after} raising the external assets of $s$ to $a^x_s + \Delta$. By separability, we define $\pre{\F}$ as the pre-network with external assets $a^x_u$ for every bank $u \in V$, i.e., the one \emph{before} raising the external assets of $s$. Then, the external assets in the difference network $\diff{\F}$ are given by $\Delta$ for $s$ and 0 for all other banks $u\neq s$. The statement follows when $\sum_{t \in T} \diff{a}_t \leq \Delta$.

We use $p^-_u$ and $p^+_u$ to denote the incoming and outgoing assets of $u$ in the clearing state of $\diff{\F}$, i.e., $p^-_u = \sum_{e \in E^-(u)} \diff{p}_e$ and $p^+_u = \sum_{e \in E^+(u)} \diff{p}_e$. Then, $\sum_{u \in V} p^-_u = \sum_{u \in V} p^+_u$, as both are the sum of all payments in the network. Dividing the banks into sinks, source $s$ and all other banks, it holds that
\begin{equation}\label{eq:non-expan}
    \sum_{u \in V\setminus (T\cup \{s\})} p^-_u + \sum_{t \in T} p^-_t + p^-_s = \sum_{u \in V\setminus (T\cup \{s\})} p^+_u + \sum_{t \in T} p^+_t + p^+_s \;.  
\end{equation}
Since $s$ is a source and the banks in $T$ are sinks, $p^-_s = 0$ and $\sum_{t \in T} p^+_t = 0$. Equation~(\ref{eq:non-expan}) simplifies to
\[
\sum_{t \in T} p^-_t = \sum_{u \in V\setminus (T\cup \{s\})} (p^+_u - p^-_u) + p^+_s \;.
\]
Every bank $u \in V \setminus (T \cup \{s\})$ has external assets of 0 and, thus, the outgoing payments of $u$ are upper bounded by the incoming ones, i.e., $p^+_u \leq p^-_u$. As a direct consequence is $\sum_{t \in T} p^-_t \leq p^+_s \leq \Delta$. Hence, $\sum_{t \in T} \diff{a}_t = \sum_{t \in T} p^-_t \leq \Delta$.
\end{proof}

\lemSourceSinkEq*

\begin{proof}
First, observe that (2) follows directly from (1). Hence, it remains to show property (1).
As in the proof of Lemma~\ref{lem:separability}, we use that the clearing state of $\F$ is the maximal fixed point of $g$ with $g_e(\vecc) = f_e(a^x_u + \sum_{e'\in E^-(u)}c_{e'})$, and the clearing state of $\HH$ is the maximal fixed point of $g^{\HH}$ with $g^{\HH}_e(\vecc^{\HH}) = f^{\HH}_e(a^x_u + \sum_{e'\in {E_{\HH}}^-(u)}c^{\HH}_{e'})$, where $u = \deb(e)$.

First, suppose $\vecp$ is the maximal clearing state of $\F$. Now consider $\vecp^{\HH}$ as in property (1) of the lemma. Then for every $e \in E$ with $\deb(e) \neq v \neq \cre(e)$, clearly $g_e(\vecp) = g^{\HH}_e(\vecp^{\HH})$ is true. Similarly for incoming edges of $e \in E^{-}(v)$ with $u = \deb(e)$, 
\[
g_e(\vecp) = f_e\left( a^x_u + \sum_{e'\in E^-(u)} p_{e'} \right) = f^{\HH}_e \left( a^x_u + \sum_{e'\in E_{\HH}^-(u)} p_{e'}^{\HH} \right) = g_e(\vecp^{\HH}).
\]
Now consider the outgoing edges $e \in E^+(v)$. It holds that
\begin{align*}
p_e = g_e(\vecp) &= f_e\left(a^x_v + \sum_{e' \in E^-(v)} p_{e'} \right) = f_e(a_v) =  f^{\HH}_e(a_{s_v}) = g_e(\vecp^{\HH})\;.
\end{align*}
Hence, $\vecp^{\HH}$ is a fixed point for the network $\HH$ and evaluates all payments as in the maximal fixed point $\vecp$ of $g$. 

Now, for the converse, suppose that $\vecp^{\HH}$ is the clearing state for $\HH$. Then, by the same equations, we obtain an equivalent fixed point of $g$. In conclusion, this proves (1).
\end{proof}

%%%%%%%%%%%%%%%%%%%%%%%%%%%%%%%%%%%%%%

The next lemma introduces a useful connection between $\HH$ and $\sig{\HH}$.

\begin{restatable}{lemma}{lemNoPosSwapOne}
\label{lem:no-pos-swap1}
Consider the financial networks $\F$ before and $\sig{\F}$ after a debt swap. If $\atvo \geq \aovo$ and $\atvt \geq \aovt$, then the following hold:

\begin{enumerate}
    \item[(1)] $\HH$ and $\sig{\HH}$ only differ in the external assets at $\svo$ and $\svt$ (and thus the clearing state).
    \item[(2)] $\HH$ is a \textit{pre-network} of $\sig{\HH}$.
    \item[(3)] $s_{v_1}$ and $s_{v_2}$ are $\{t_{u_1}, t_{u_2}, \tvo, \tvt\}$\textit{-linear} in $\HH$ on intervals $[\aovo,\, \atvo]$ and $[\aovt,\, \atvt]$, respectively. 
\end{enumerate}
\end{restatable}

%\lemNoPosSwapOne*

\begin{proof}
By definition, the underlying graphs of $\F$ and $\sig{\F}$ only differ in the liabilities of $u_1$ and $u_2$ to $v_1$ and $v_2$. When applying source-sink equivalency to these edges, all other edges remain the same. The resulting networks $\G$ and $\sig{\G}$ only differ in external assets of $v_1$ and $v_2$. This still holds when applying source-sink equivalency to $v_1$ and $v_2$ to $\G$ and $\sig{\G}$. Hence, (1) is satisfied. By definition, (1) directly implies (2).

Towards (3), consider the pre-network $\HH$. Suppose the external assets at the sources $\svo$ and $\svt$ are increased by $\atvo - \aovo$ and $\atvt - \aovt$, respectively. In total, the external assets of both banks are raised by $(\atvo + \atvt) - (\aovo + \aovt).$

By (1), every bank now has the same external assets and liabilities as in $\sig{\HH}$ and, thus, the clearing state is $\sig{\vecp}$. For $\sig{\HH}$ we use the notation $\sig{e}_1 = (u_1,t_{u_2})$ and $\sig{e}_2 = (u_2,t_{u_1})$ for the swapped edges. Then, the total assets of the sinks $t_{u_1}$, $t_{u_2}$, $\tvo$, and $\tvt$ increase by a total amount of
\begin{align*}
&(\sig{a}_{t_{u_1}} + \sig{a}_{t_{u_2}} + \sig{a}_{t_{v_1}} + \sig{a}_{t_{v_2}}) - (a_{t_{u_1}} + a_{t_{u_2}} + a_{t_{v_1}} + a_{t_{v_2}})\\
=\;& \left( \sig{p}_{\sig{e}_2} + \sig{p}_{\sig{e}_1} + (\atvo - a^x_{v_1} - \sig{p}_{\sig{e}_2}) + (\atvt - a^x_{v_2} - \sig{p}_{\sig{e}_1}) \right) \\
& \hspace{0.5cm} - \left( p_{e_1} + p_{e_2} + (\aovo - a^x_{v_1} - p_{e_1}) + (\aovt - a^x_{v_2} - p_{e_2})  \right) \\
=\;& (\atvo + \atvt) - (\aovo + \aovt).
\end{align*}
Consequently, sources $\svo$ and $\svt$ are $\{t_{u_1}, t_{u_2}, \tvo, \tvt\}$-linear on $[\aovo,\,\, \atvo]$ and $[\aovt,\,\, \atvt]$, respectively.
\end{proof}

\lemNoPosSwapTwo*

\begin{proof}
(1) directly follows from $\{t_{u_1}, t_{u_2}, \tvo, \tvt\}$-linearity of $\svo$ and $\svt$ on intervals $[\aovo, \atvo]$ and $[\aovt, \atvt]$ (Lemma~\ref{lem:no-pos-swap1}).
For (2), assume that $\delta_{\svo}> 0$ or $\delta_{\svt}>0$ and $\delta_{\svo} = \Delta_{\tvo} + \Delta_{t_{u_1}}$. Then, by (1) it follows that $\delta_{\svt} = \Delta_{\tvt} + \Delta_{t_{u_2}}$. 

Intuitively, the additional external assets of $v_1$ are forwarded through the network and, eventually, are payed back to $v_1$. The same holds for the additional external assets of $v_2$. To determine the external assets of $v_1$ in $\F_{\Delta}$, first, observe that the incoming assets of $v_1$ are at least the incoming assets of $v_1$ in $\F$ of $a_{v_1} - a^x_{v_1}$. In $\F_{\Delta}$, the bank receives additional payments of $\Delta_{\tvo} + \Delta_{t_{u_1}}$. Furthermore, the raised external assets could enable payments in new cycles where all payments from $v_1$ towards the cycle are forwarded back to $v_1$. 

Formally, recall that the total assets of $s_{v_1}$ in $\HH_{\Delta}$ equal the total assets of $v_1$ in $\F_{\Delta}$. The incoming payments of $v_1$ in $\F_{\Delta}$ are the sum of total assets of $t_{u_1}$ and $t_{v_1}$ in $\HH_{\Delta}$. Since $\HH$ is a pre-network of $\HH_{\Delta}$ with clearing state equivalent to $\F$, this sum is given by $(a_v - a_v^x) + \Delta_{t_{v_1}} + \Delta_{t_{u_1}}$. Now the external assets of $v_1$ in $\F_{\Delta}$ are the total assets of $s_{v_1}$ minus the sum of total assets of $t_{u_1}$ and $t_{v_1}$ in $\HH_{\Delta}$, i.e.,
\begin{align*}
    & a_{v_1} + \delta_{v_1} - \left((a_{v_1} - a^x_{v_1}) + \Delta_{\tvo} + \Delta_{t_{u_1}}\right)\\
    =\;& a^x_{v_1} + \delta_{v_1} - \Delta_{\tvo} - \Delta_{t_{u_1}}\\
    =\;& a^x_{v_1} 
\end{align*}
where the last equality follows by assumption.
Analogously, $v_2$ has external assets of $a^x_{v_2}$ in $\F_{\Delta}$. Hence, every bank has the same external assets in $\F$ and in $\F_{\Delta}$. However, since $a_{v_1} + \delta_{v_1} > a_{v_1}$ or $a_{v_2}+ \delta_{v_2} > a_{v_2}$ the payments in the clearing state in $\F_{\Delta}$ are strictly higher than in $\F$. This contradicts the fact that the clearing state $\vecp$ is the maximal fixed point. With the same argument, $\delta_{\svt} \neq \Delta_{\tvt} + \Delta_{t_{u_2}}$ follows.
\end{proof}

\thmNoPositive*

\begin{proof}
For a financial network $\F$, suppose the debt swap of $e_1$ and $e_2$ is positive, i.e., $\atvo > \aovo$ and $\atvt > \atvo$. We set $\delta_{v_1}$ and $\delta_{v_2}$ to the increase in total assets of $v_1$ and $v_2$ due to the swap, i.e., $\delta_{v_1} \coloneqq \atvo - \aovo$ and $\delta_{v_2} \coloneqq \atvt - \aovt$. Now consider the financial network $\HH$ corresponding to $\F$ and raise the external assets of $s_{v_1}$ by at most $\delta_{v_1}$ and external assets of $s_{v_2}$ by at most $\delta_{v_2}$. Let $g_{t_{u_1}}: [0, \delta_{v_1}] \times [0, \delta_{v_2}] \to [0, \infty]$ be the increase in total assets at $t_{u_1}$ as a function of the increase in external assets at $\svo$ and $\svt$ in the clearing state of network $\HH$. We define $g_{t_{v_1}}$ similarly. The payment functions of all banks are non-decreasing and continuous, and so are $g_{t_{v_1}}$ and $g_{t_{u_1}}$. In particular, $g_{t_{v_1}} (0,0) = g_{t_{u_1}}(0,0) = 0$. 

The sum $g_1 = g_{t_{u_1}} + g_{t_{v_1}}$ is non-decreasing and continuous as well, with $g_1(0, 0) = 0$. By non-expansivity, we have that $g_1(x,y) \le x + y$ for all $(x,y) \in [0, \delta_{v_1}] \times [0, \delta_{v_2}]$. We will now show that whenever the debt swap of $e_1$ and $e_2$ is positive, there exists an $(x,y) \in [0, \delta_{v_1}] \times [0, \delta_{v_2}]$ with $(x,y)>(0,0)$, i.e., $x>0$ and $y>0$, such that $g_1(x,y) = x$. This is a contradiction to Lemma~\ref{lem:no-pos-swap2} (2). 

We proceed in three steps.
\begin{enumerate}[(i)]
\item $g_1(x,0)<x$ for all $x \in (0,\delta_{v_1}]$: \\
Since $g_1(x, y) \leq x + y$, it follows that $g_1(x, 0) \leq x$, for all $x>0$. Now $g_1(x,0)=x$ for some $x > 0$ would imply $\{\tvo, t_{u_1}\}$-linearity of $s_{v_1}$ (Lemma~\ref{lem:no-pos-swap2} (1)) and a contractiction to Lemma~\ref{lem:no-pos-swap2} (2).

\item  $g_1(c,y) > c$ for some $c\in (0, \delta_{v_1}], y \in (0, \delta_{v_2}]$: \\
Fix a value $y \in (0, \delta_{v_2}]$ with $y \leq \delta_{v_1}$. Then clearly $g_1(0,y)\geq 0$. Moreover, $g_1(0,y)>0$ since otherwise $s_{v_2}$ would be $\{\tvt,t_{u_2}\}$-linear (Lemma~\ref{lem:no-pos-swap2} (1)) and we obtain a contradiction to Lemma~\ref{lem:no-pos-swap2} (2). Note that $g_1(0,y) \le y \le \delta_{v_1}$. Since $g_1(0,y)>0$ and $g_1$ is monotone, we obtain $g_1\left(g_1(0, y), y\right) \geq g_1(0, y)$.
Using the notation $c \coloneqq g_1(0,y)$, we have $0 < c \leq \delta_{v_1}$ and $g_1(c,y) \leq c$. Since $g_1(c,y) = c$ implies a contradiction to Lemma~\ref{lem:no-pos-swap2} (2), we must have $g_1(c,y) > c$.

\item  $g_1(x,y)=x$ for some $(x,y) \in (0,\delta_{v_1}] \times (0,\delta_{v_2}]$: \\
(i) and (ii) imply that there exists a pair $(c,y) \in [0,\delta_{v_1}] \times [0,\delta_{v_2}]$ such that $g_1(c,0) < c$ and $g_1(c,y)>c$. By the intermediate value theorem, there exists a $d \in [0,y]$ satisfying $g_1(c,d)=c$. Observe that in particular $y > 0$ holds true. \qedhere
\end{enumerate}
\end{proof}

The theorem substantially generalizes an earlier result for proportional payments~\cite{PappDebtSwapping}. Let us note that the proof in the full version of~\cite{PappDebtSwapping} relies on a Lemma B.1 which is incorrect. The lemma claims the following: For two financial networks $\F^1$ and $\F^2$ with proportional payments on the same set of nodes, suppose there exists one source node $s$ with external assets $a^x_s$ while $a^x_u=0$ for all other banks in $\F^1$ and $\F^2$. If $l^{\F^1}_e \leq l^{\F^2}_e$ for every edge, then $a^{\F^1} \leq a^{\F^2}$ follows. 

As a counterexample to this claim, consider a financial network over the three banks $s,w_1$ and $w_2$ and edges $(s,w_1)$ and $(s,w_2)$ with $l_{(s,w_1)}=1$ and $l_{(s,w_2)}=2$. Let the external assets be given by $a^x_s = 2$ and $a^x_{w_1}=a^x_{w_2}=0$. Then, $s$ pays $\frac{1}{1+2}\cdot 2= \frac{2}{3}$ towards $w_1$ and $\frac{2}{1+2}\cdot 2= \frac{4}{3}$ towards $w_2$. Now assume the edge weight $l_{(s,w_1)}$ is raised to $2$. Then, the payments from $s$ to $w_1$ as well as $w_2$ are changed to $\frac{2}{2+2}\cdot 2 = 1 < \frac 43$.

\propNoSwapPosPaym*

\begin{proof}
For a given network $\F$ consider the corresponding network $\G$ after applying source-sink equivalency to the edges $e_1$ and $e_2$. This procedure adds the new banks $t_{u_1},t_{u_2}$ and changes the edges $e_1$ and $e_2$ from $(u_1,v_1)$ and $(u_2,v_2)$ to $(u_1,t_{u_1})$ and $(u_2,t_{u_2})$, respectively. Additionally, the external assets of $v_1$ and $v_2$ are updated to $a^x_{v_1}+p_{e_1}$ and $a^x_{v_2}+p_{e_2}$, respectively. Analogously, define $\sig{\G}$ for $\sig{\F}$.
Observe that $\G$ and $\sig{\G}$ only differ in the external assets of $v_1$ and $v_2$. Hence, $\G$ is a pre-network of $\sig{\G}$ and, by monotonicity of the allocation rule, both $v_1$ and $v_2$ strictly profit from the edge swap. This is a contradiction to Theorem~\ref{thm:no-pos-swap}.
\end{proof}

\propSemiSwapPaym*

\begin{proof}
Proposition~\ref{prop:no-swap-pos-payments} shows that $\sig{p}_{e_1} > p_{e_1}$ and $\sig{p}_{e_2} > p_{e_2}$ can never be satisfied simultaneously. If $\sig{p}_{e_1} \leq p_{e_1}$ and $\sig{p}_{e_2} \leq p_{e_2}$, then the same source-sink transformation as in the proof of Proposition~\ref{prop:no-swap-pos-payments} shows that the network after swapping $\sig{\G}$ is a pre-network of $\G$. By Corollary~\ref{cor:pre-network-leq-clearing} such a swap cannot be semi-positive. Consequently, either $\sig{p}_{e_1} > p_{e_1}$ and $\sig{p}_{e_2} \leq p_{e_2}$, or $\sig{p}_{e_2} > p_{e_2}$ and $\sig{p}_{e_1} \leq p_{e_1}$.
\end{proof}

\thmPareto*

\begin{proof}
First, let the debt swap of $e_1$ and $e_2$ be Pareto-improving. For a contradiction, assume it is not semi-positive. Then, w.l.o.g.\ the following two cases can occur: (1) $\atvo < \aovo$ and $\atvt > \aovt$, or (2) $\atvo \leq \aovo$ and $\atvt \leq \aovt$. Towards (1), it is easy to see that $v_1$ is harmed by the swap. Hence, it is not Pareto-improving. For (2), consider the two networks $\HH$ and $\sig{\HH}$ after applying source-sink equivalency to the swapping edges and $v_1$ and $v_2$ in $\F$ and $\sig{\F}$. Since $\HH$ and $\sig{\HH}$ only differ in the external assets of $v_1$ and $v_2$, $\HH$ is a pre-network of $\sig{\HH}$. Consequently, by Corollary~\ref{cor:pre-network-leq-clearing}, the total assets of all banks in $\HH$ are at least as high as those $\sig{\HH}$. Hence, no bank strictly profits by the debt swap, a contradiction.

For the other direction, suppose the debt swap of $e_1$ and $e_2$ is semi-positive and w.l.o.g.\ $\atvo > \aovo$ and $\atvt = \aovt$. Then, by Lemma~\ref{lem:no-pos-swap1},  $\HH$ is a pre-network of $\sig{\HH}$, and thus, by Corollary~\ref{cor:pre-network-leq-clearing}, the total assets of all banks after the swap are at least as high as those before the swap. Thus, the debt swap is a Pareto-improvement.
\end{proof}

%%%%%%%%%%%%%%%%%%%%%%%%%%%%%%%%%%%%%%%%%%%

\end{document}